\documentclass{article}
\usepackage[left=1in,right=1in,top=1in,bottom=1in]{geometry}
\usepackage{amsfonts,amsmath,mathrsfs,bbm,amsthm, amssymb, enumitem,comment,apptools}
\usepackage{xcolor}
\usepackage{authblk}
\usepackage{appendix}
\usepackage{enumitem}
\usepackage{graphicx}
\usepackage[hidelinks]{hyperref}
\usepackage{url}
\usepackage{tikz-cd}
\usepackage{bbm}
\usepackage{orcidlink}

\DeclareMathOperator*{\esssup}{\mathrm{esssup}}

\newcommand{\Tr}{\text{Tr}}
\DeclareMathOperator{\id}{id}

\newcommand{\calN}{\mathcal{N}^{\mathfrak{h}}_F}
\newcommand{\one}{\mathbbm{1}}
\DeclareMathOperator{\dom}{dom}

\newcommand{\Znu}{\mathbb{Z}^\nu}
\newcommand{\C}{\mathbb C}
\newcommand{\Z}{\mathbb Z}
\newcommand{\A}{\mathcal{A}}
\newcommand{\B}{\mathcal{B}}
\newcommand{\M}{\mathbb{M}}
\newcommand{\N}{\mathbb N}
\newcommand{\R}{\mathbb R}

\DeclareMathOperator{\diam}{diam}

\renewcommand{\P}{\mathbb{P}}

\DeclareMathOperator{\loc}{loc}

\newcommand{\aloc}{\mathcal{A}_{\mathbb{Z}^\nu}^{\loc}}
\newcommand{\<}{\langle}
\renewcommand{\>}{\rangle}

\let\emptyset\varnothing

\newcommand{\Alambda}{\mathcal{A}_{\Lambda}}
\DeclareMathOperator{\spec}{spec}

\newcommand{\h}{\mathfrak{h}}
\newtheorem{thm}{Theorem}[section]
\newtheorem{thmx}{Theorem}

\newtheorem{prop}[thm]{Proposition}
\newtheorem{lem}[thm]{Lemma}
\newtheorem{cor}[thm]{Corollary}
\newtheorem*{thm*}{Theorem}
\theoremstyle{definition}
\newtheorem{define}[thm]{Definition}

\newtheorem{rmk}[thm]{Remark}

\newtheorem{note}[thm]{Notation}

\numberwithin{equation}{section}
\begin{document}
\title{\huge\textbf{Disordered Ground States of Ergodic Quantum Spin Systems}}
\author{Eric B. Roon\orcidlink{0009-0000-7566-5734}\thanks{\url{rooneric@msu.edu}}\quad\&\quad Jeffrey H. Schenker\orcidlink{0000-0002-1171-7977}\thanks{\url{schenke6 @ msu.edu}}\\Department of Mathematics\\ Michigan State University\\ East Lansing, MI., U.S.A.}
\date{\today}                     
\setcounter{Maxaffil}{0}
\renewcommand\Affilfont{\itshape\small}
\maketitle

\begin{abstract}
In this letter, we fill a hole in the existing literature about disordered quantum spin systems generated by a random local interaction $\{\mathfrak{h}(Z)\}_{Z\Subset \Znu}$ satisfying a statistical version of translation invariance. We show such systems always have disordered ground states in the thermodynamic limit with the same symmetry. A key tool we use is a disordered version of the Lieb-Robinson bounds, which hold almost surely under mild conditions on $\mathfrak{h}$. Along the way, we formalize the notion of a random state on a $C^*$-algebra and prove a weak-$\ast$ version of the Riesz-Markov-Kakutani theorem, which seems not to have been recorded in the vector measures literature. As a consequence of the existence of the aforementioned disordered ground states, we show that the spectrum of the GNS Hamiltonain associated to the bulk dynamics is deterministic with respect to the disorder. 
\end{abstract}

\tableofcontents
\section*{Acknowledgements}
EBR gratefully acknowledges the Harvard Center for Mathematical Sciences and Applications (CMSA) and the Technische Universit\"{a}t M\"{un}chen for their hospitality during conferences he attended in June and September 2025 (respectively) where some of this work was completed. EBR is grateful to Drs. Martin Fraas, Lucas Hall and Bruno Nachtergaele for the helpful conversations. This work was partially supported by the National Science Foundation under Grant No. 2153946.  

\section{Introduction}
Models of the effects of disorder on the dynamics of quantum spin chains has been a topic of interest since at least the late 1970's \cite{Pastur, PasturFigotin}. We will not attempt to review all of the literature, but let us draw the reader's attention to a few recent, important works  \cite{Abdul-Rahman_et_al, AizenmannWarzel, Anderson,  Baldwin, ElgartKlein24_2, ElgartKlein24_1, ElgartKlein25,  Stolz_ALoc, Stolz_MBL}. Among these, models of quantum spin systems with intrinsic disorder have been analyzed due to interest in phenomena like many-body localization or zero-velocity propagation estimates for spin chains \cite{ElgartKlein25, Imbrie, NachtergaeleReschke, Stolz_MBL} to name a few. 

On the other hand, among the most important properties of a given quantum spin model is whether or not a there is a spectral gap above its ground state. For example, the existence of a spectral gap implies exponential decay of spatial correlations in the ground state \cite{HastingsKoma, NachtergaeleSims}. As such, a lot of work over the past thirty years has been invested in studying the consequences of the presence of a bulk gap \cite{NachtergaeleSims}, whether or not the global property of having a bulk gap is stable under local perturbations \cite{BravyiHastingsMichalakis, MichalakisZwolak, NachtergaeleSimsYoung, NachtergaeleSimsYoung_2, NachtergaeleSimsYoung_bulk}, how one may tell from local data if a system has a global gap \cite{Knabe, Lemm_ATMP, Nachteragele}, etc. 

Recently, analysis carried out in \cite{JauslinLemm, LancienPerezGarcia} assessed how typical a spectral gap above the ground state is for translation invariant matrix product state systems chosen randomly from a distribution and iterated across the lattice $\Z$. In our recent pre-print \cite{RoonSchenker}, the authors study an explicit disordered deformation of the AKLT model \cite{AKLT, FannesNachtergaeleWerner} which we show admits a nearest-neighbor parent Hamiltonian for a pure state where a kind of \emph{statistical} translation invariance is present. Using the robust machinery of quasi-locality estimates, we were able to show that the bulk spectral gap of our model is a) deterministic; and b) closes almost surely provided some technical assumptions on the disorder are satisfied. 

In the random Schr\"odinger operators literature, it has been known since the early 1980's that ergodically implemented disorder with translation symmetry results in a deterministic spectrum for the total Hamiltonian \cite{KirschMartinelli, Lentz_et_al, Pastur75}. The purpose of this note is to extend this result to the quantum spin system setting. 

To wit, we consider a countably generated probability space $(\Omega, \mathcal{F}, \P)$ and a random local interaction $\{h^{\Lambda}_\omega \colon \Lambda \subset \Znu \text{ finite}\}$ whose norm decays sufficiently fast as $\diam(\Lambda) \to \infty$. More precisely, we utilize the $F$-norm for quasilocal interactions introduced in \cite{NachtergaeleOgataSims} (we discuss the $F$-norm in section~\ref{sec:Quasilocality} below; see also any of \cite{NachtergaeleSimsYoung, NachtergaeleSimsYoung_2, NachtergaeleSimsYoung_bulk}). We are concerned with \textbf{ergodic interactions} in the sense that $\Znu$ acts on $\Omega$ by a family of measure preserving ergodic maps $\{\vartheta_x\}_{x\in \Znu}$ so that under translation,
\begin{equation}\label{eqn:translation-covariance}
    \tau_{x} h^{\Lambda}(\omega) = h^{\Lambda +x}(\vartheta_x \omega),\,\quad \P-\text{a.e.}\, \omega\,, \text{ all }x\in \Znu.
\end{equation}
\begin{thmx}[Theorem~\ref{thm:exist} Informal]\label{thmx:exist}
	If $\Znu$ is equipped with a translation-invariant metric, and  $\{h^{\Lambda}_\omega\}$ is an ergodic interaction as in~(\ref{eqn:translation-covariance}) above, then there exist ergodic disordered bulk states $\psi_\omega$ satisfying
\begin{equation}
    \psi_\omega \circ \tau_x = \psi_{\vartheta_x \omega},\, \P-a.e.\, \omega, \qquad\forall x \in \Znu\,.
\end{equation}
\end{thmx} We note that by bulk ground state, we mean that $\psi_\omega$ is weakly-$\ast$ measurable and satisfies the $C^*$-algebraic ground state condition $\omega$-point-wise (cf inequality~(\ref{eqn:C_star_ground_state}) below). In the preprints \cite{EkbladMoreno-NadalesRoonSchenker, RoonSchenker}, the property~(\ref{eqn:translation-covariance}) is referred to as ``translation covariance." Here, we use the term ``ergodic," in alignment with with how it is used in the Schr\"odinger operators literature. We note that we are therefore using the term ergodic to mean two different things in this article: one, the property of the automorphisms $\vartheta_x$ acting on the probability space; and, two, the property~(\ref{eqn:translation-covariance}). It will be clear from context which definition we mean. 

A common example of an interaction which satisfies~(\ref{eqn:translation-covariance}) is given by a nearest-neighbor Hamiltonian interaction $h_{j, j+1}$ with a random transverse field  \[
	H^{[0,L]} = \sum_{j=0}^{L-1} h_{j, j+1} + \sum_{j=0}^L \lambda_j \sigma^z_j\,,
\] where the $\lambda_j$ are IID random variables, as is studied in (e.g.) \cite{Abdul-Rahman_et_al, Baldwin, ElgartKlein24_1, HamzaSimsStolz, Stolz_MBL}. The ergodic map in this context is the ergodic shift of the sequence space of outcomes for the $\lambda_j$. A particular case of such a nearest-neighbor Hamiltonian is the random XY-spin chain, where $h_{j, j+1} = \mu_j ((1+\gamma_j) \sigma_j^x\sigma_{j+1}^x + (1-\gamma_j) \sigma_{j}^y\sigma_{j+1}^y)$ (where $\mu, \gamma$ are deterministic parameters)\cite{Abdul-Rahman_et_al}. Random spin chains has been mathematically studied by various authors in the context of many-body-localization (see the nonexhaustive list \cite{Abdul-Rahman_et_al,  ElgartKlein24_2, ElgartKlein24_1, ElgartKlein_2026,  ElgartKlein25, Imbrie, NachtergaeleReschke, SimsStolz, Stolz_MBL}). We emphasize that MBL and related phenomena are \emph{not} the topic of this letter. Rather, we are concerned with providing a mathematical framework that includes these kinds of models as well as the disordered deformation of the AKLT model we study in \cite{RoonSchenker}.

One commonality between these seemingly disparate models is that the Hamiltonians here are known to have ergodic (in the disorder) ground states, but the methods used to prove such results are highly model-specific. The import of Theorem~\ref{thmx:exist} is that such ergodicity of the disordered local interactions implies there are ergodic bulk ground state sectors.

The usual proof to show that there is a translation-invariant ground state is to examine the weak-$\ast$ limit points of spatial averages of any ground state \cite{BratteliRobinsonII}. However in our case, the subsequence along which those limits exist may depend on the disorder, and may change after applying translation. Therefore it is necessary to argue from a functional analytic perspective. In particular, we need to assume that the probability space $(\Omega, \mathcal{F}, \P)$ is countably generated in order for the weak-$\ast$ topology in an appropriate dual space to be sequentially compact in the unit ball. However, this assumption encompasses the so-called standard probability spaces, and models of IID and quasiperiodic disorder. A key tool is a kind of $C^*$-version of the Riesz-Markov-Kakutani theorem (see Proposition~\ref{prop:w*RND} below).

For our second main result, we show that a version of the Pastur-Kirsch-Martinelli theorem holds for the spectrum of the GNS Hamiltonian with respect to an ergodic ground state.
\begin{thmx}[Corollary~\ref{thm:Bulk_gap_const} Informal]\label{thmx:spec}
    Let $\{h^{\Lambda}_\omega \colon \Lambda \Subset \Znu\}$ be a local ergodic interaction with sufficient spatial decay. Suppose $\gamma_\omega$ is a translation-covariant ground state, and let $(\mathcal{H}_\omega, \pi_\omega, \Psi_\omega)$ be the associated GNS representations. Then the thermodynamic limit $(\alpha^{\Znu}_{t;\omega})_{t\in \R}$ passes to a strongly continuous unitary group on $\mathcal{H}_\omega$ generated by a self-adjoint operator $H_\omega$. Furthermore, there are unitaries $\{U_x\colon x\in \Znu\}$ so that 
    	\[
		H_{\vartheta_x \omega} = U^*_x H_\omega U_x\,.
	\] Whence, the spectrum of the GNS generator for $\alpha^{\Znu}_{t;\omega}$ in a covariant ground state is deterministic almost surely.
\end{thmx}

Interestingly, the condition of local ergodicity as in~(\ref{eqn:translation-covariance}) appeared in the quantum spin system context as early as 1978 in the work of Pastur and Figotin \cite{PasturFigotin}. Pastur \cite{Pastur} goes on to analyze thermodynamic limits of the classical random Ising model. Both Pastur and Figotin \cite{PasturFigotin} and Pastur \cite{Pastur} show that certain quantities (e.g. the free energy) associated to a local ergodic interaction are deterministic in the thermodynamic limit. Something similar is true in our context for the Lieb-Robinson velocity (see Corollary~\ref{cor:N-norm_const,thermo_covariance}), namely that the so-called F-norm of an ergodic local Hamiltonian is deterministic. Thus when finite, this provides an essential upper bound to the Lieb-Robinson velocity.


\section{Preliminaries}
Below we set notation and state some important background information that the reader may find convenient. The goal of this section is to point out facts which may be already known by experts in the community, but whose statements will make this letter more understandable. We provide references and proofs when necessary. 
\subsection{Quantum Spin Systems}
Our setting is the integer lattice $\Znu$ with dimension $\nu\in \N$. To each site $x\in \Znu$ we associate a replica of the $k\times k$ matrices $\mathcal{A}_x\cong \M_k$. By $\mathscr{P}_0(\Znu)$ we mean the set of finite subsets of $\Znu$. When it is convenient, we will write $Z\Subset \Znu$ to indicate $Z\in \mathscr{P}_0(\Znu)$. Given $\Lambda \Subset \Znu$ one forms the \emph{local algebra (of observables)} by taking the tensor product $\Alambda = \bigotimes_{x\in \Lambda} \mathcal{A}_x\,.$ Whenever $\Lambda_0 \subset \Lambda$ are finite, there is a canonical $*$-homomorphism $\iota_{\Lambda_0, \Lambda} : \A_{\Lambda_0} \hookrightarrow \A_{\Lambda}$ given by taking tensor products with the identity:
\[
    \iota_{\Lambda_0, \Lambda}(a_{\Lambda_0}) = a_{\Lambda_0} \otimes \one_{\Lambda \setminus \Lambda_0}\,,
\] and extending linearly. Importantly, $\iota_{\Lambda_0,\Lambda}$ is an isometry, thus one may identify $\A_{\Lambda_0}$ with the $C^*$-sub-algebra $\iota_{\Lambda_0, \Lambda}(\A_{\Lambda_0})$ of $\A_{\Lambda}$. In this way the family $\{\A_{\Lambda} \colon \Lambda \in \mathscr{P}_0(\Znu)\}$ forms a net. One may further define the \emph{quasi-local algebra (of observables)} as the norm completion of the union of the local algebras, 
\[
    \A_{\Znu} := \overline{\bigcup_{\Lambda \Subset \Znu} \A_{\Lambda}}^{\|\cdot \|}\,.
\] Note that there are then canonical $*$-homomorphism inclusions $\iota_{\Lambda}:\A_{\Lambda} \to \A_{\Znu}$ and so one may regard the local algebras once again as subalgebras of $\A_{\Znu}$. For convenience we write $\A_{\Znu}^{\loc}$ for the dense $*$-subalgebra of $\A_{\Znu}$ corresponding to the union of the local algebras. Following the convention of \cite{NachtergaeleSimsYoung, NachtergaeleSimsYoung_2, NachtergaeleSimsYoung_bulk}, we say that a sequence of finite volumes $\{\Lambda_n\}_{n=1}^\infty$ is an \emph{increasing and adsorbing sequence (IAS)} if $\Lambda_n \subset \Lambda_{n+1}$ for all $n\ge 1$ and $\bigcup \Lambda_n = \Znu$. We summarize some key properties below. 

\begin{prop}
    Suppose $\A_{\Znu}$ is the quasilocal algebra associated to one-site algebras $\A_{\{x\}}\cong \M_k$. The following hold: 
    \begin{enumerate}
        \item Let $(\Lambda_n)_{n=1}^\infty$ be any IAS. Then, $\A_{\Znu}^{\loc}$ is an increasing union of the increasing sequence of algebras $\A_{\Lambda_n}$. 
        \item In particular, $\A_{\Znu}$ is norm separable and contains a unit. Furthermore, there is a unique state $\rho\in \A_{\Znu}^*$ so that $\rho(ab) = \rho(ba)$ for all $a,b\in \A_{\Znu}$. We call $\rho$ the \emph{tracial state}.
        \item For any IAS $(\Lambda_n)_{n=1}^\infty$, there exist $\rho$-invariant unital completely positive maps $\mathcal{E}_n:\A_{\Znu} \to \A_{\Lambda_n}$ so that for any $a\in \A_{\Znu}$, the limit $\|a - \mathcal{E}_n(a)\|_{\A_{\Znu}} \to 0$ as $n\to \infty$. 
    \end{enumerate}
\end{prop} The last item is a consequence of being able to take partial traces in finite volumes (cf. \cite{NachtergaeleScholzWerner, NachtergaeleSimsYoung}). In the operator algebras literature,  the third property is referred to as \emph{nuclearity}, which holds for a much larger class of $C^*$-algebras than we consider here (see Proposition~\ref{prop:w*_simple_functions} below for our humble application of this more general concept). 

Recall that $\Znu$ is canonically an abelian group under point-wise addition, and it acts on itself via translation isomorphisms: 
    \[
        \tau_k(x_1, \dots, x_\nu) := (x_1 +k_1, \dots, x_\nu +k_\nu)\,.
    \] One clearly has $\tau_{k}\circ \tau_j = \tau_{k+j}$. These translations induce automorphisms of the quasilocal algebra by translating the indices of pure tensors via
    \[
        \tau_k(a_{x_1} \otimes \cdots \otimes a_{x_n}) := a_{x_1 +k}\otimes \cdots \otimes a_{x_n +k}\,.
    \] These $\tau_k$ are norm-preserving and therefore extend to $\ast$-automorphisms of $\A_{\Znu}$.

\subsection{Random states on Quasilocal Algebras.}
We will briefly recall two technical results which will be useful for addressing one of the central questions of this letter which is measurability of weak-$\ast$ limits of states in the quasilocal algebra. For us, a \textbf{random state} is a function $\phi:\Omega \to \A_{\Znu}^*$ which satisfies the following: 
\begin{enumerate}
    \item for all $a\in \A_{\Znu}$, the numerical mapping $\omega \mapsto \phi_\omega(a)$ is Borel measurable;
    \item for all $a\in \A_{\Znu}$ the probability $\P[\phi_\omega(a^*a) \ge 0] =1$;
    \item one has $\phi_\omega(\one) =1$ almost surely. 
\end{enumerate} 
Recall that when $\A$ is a separable $C^*$-algebra, the weak-$\ast$ topology  on the unit ball in $\A^*$  is metrizable \cite[Theorem 3.16]{Rudin} and compact \cite[Theorem V.3.1]{Conway_FA}. Since compact metric spaces are complete, this makes $(\mathcal{S}(A), w*)$ a compact Polish space which is also sequentially compact. We emphasize that the unit ball of $\A^*$ is not generally sequentially compact: this occurs because of the existence of a metric. The following lemma tells us how to lift sequential compactness to weak-$\ast$ ``disordered" sequential compactness in $\A^*$.

Our first lemma allows us to guarantee that weak-$\ast$ limits of random finite volume states are measurable which is a crucial ingredient in the proof that disordered ground states exist in Lemma~\ref{lem:gs} below. 

\begin{lem}\label{lem:hitting_time_compact}
    Let $(\Omega, \mathcal{F}, \P)$ and let $\A$ be a separable unital $C^*$ algebra. Let $\mathscr{E}$ be the closed unit ball of $\A^*$ equipped with the weak-$\ast$ topology, which we induce via a choice of metric $d$. If $(\Phi_n:\Omega \to \mathscr{E})_{n=1}^\infty$ is a sequence of weak-$\ast$ measurable random variables with the property that the norm function $\|\, \|\Phi_n(\omega)\|_{\A^*}\|_{L^\infty(\P)}\le 1$ for all $n$. Then, there exists a sequence of $\N$-valued random variables $(\sigma_n: \Omega \to \mathbb{N})_{n=1}^\infty$ which are finite almost surely, and a weakly-$\ast$-measurable function $\Psi:\Omega \to \mathscr{E}$ so that $\lim_{n\to \infty} \Phi_{\sigma_n} = \Psi$ almost surely.
\end{lem}
\begin{proof}
    Consider the closed cover $\mathcal{U}_1:=\{b_\psi(2^{-1})\colon \psi \in \mathscr{E}\}$ of $\mathscr{E}$ by metric balls of radius $1/2$. By compactness, there is a finite sub-cover $\{b_{\psi_1}(2^{-1}), \dots, b_{\psi_k}(2^{-1})\}$ of $\mathcal{U}$. Then, there is some $k_1(\omega)$ which is finite almost surely so that $B_{\psi_{k_1(\omega)}}(2^{-1})$ contains infinitely many members of the sequence $\Phi_n(\omega)$. 
    
    Let $\sigma_{1}(\omega) = \inf\{n\colon \Phi_n(\omega) \in B_{\psi_{k_1(\omega)}}(2^{-1})\}$. Now, consider the closed cover $\mathcal{U}_2 = \{b_{\psi'}(2^{-2})\}$ of $b_{\psi_{k_1(\omega)}}(2^{-1})$ by closed balls of radius $2^{-2}$. By the same argument, there is some $k_2(\omega)$ so that infinitely many members of the sequence $\Phi_n(\omega)$ belong to $b_{k_2(\omega)}(2^{-2})$ when $n\ge \sigma_n(\omega)$. Let $\sigma_2(\omega) := \inf\{ n> \sigma_1(\omega)  \colon \varphi_n(\omega) \in b_{\psi_2'(\omega)}(2^{-2}) \}$. 
    
    Proceed inductively to produce a measurable sequence of hitting times $\{\sigma_1, \sigma_2, \dots\}$ which is finite almost surely (once again by compactness). For each realization of the disorder, we have produced a sequence of nested closed sets in a compact metric space $b_{\psi_{k_1}(\omega)}(2^{-1}) \supset b_{\psi_{k_2}'(\omega)}(2^{-2}) \supset \cdots$. By basic topology, this tells us that there is a single point, $\Phi(\omega)$ in their common intersection and furthermore that $\lim_{k\to \infty} d(\Phi_{\sigma_k(\omega)}(\omega), \Phi(\omega))\le\lim_{k\to \infty} 2^{-k+1} = 0$. Whence $\Phi(\omega)$ is weakly-$\ast$-measurable. 
\end{proof}

Turning to the next result, we remark that in order to prove that \emph{ergodic} bulk ground states exist, we found it necessary to argue at the level of vector measures. This is because the weak-$\ast$ limits of a shifted sequence of ground states depend on a random time $\sigma_n$ which need not vary ergodically. Proposition~\ref{prop:w*RND} allows us to get around this in the proof of Theorem~\ref{thmx:exist} below. We want to enlarge the the class of linear functionals we consider to encompass all linear functionals on the Bochner space $L^1(\A_{\Znu})$ of measurable $\A_{\Znu}$-valued functions with integrable operator norm functions. Given such an $A:\Omega \to \A_{\Znu}$ with $\int_\Omega \|A_\omega\| d\P <\infty$, one can define a linear functional $\Phi$ by integrating the composition with $A$ as follows. 
\[
    \Phi(A) := \int_\Omega \phi_\omega(A_\omega) d\P(\omega)\,.
\] This defines a vector measure into $\A_{\Znu}^*$ which is absolutely continuous with respect to $d\P$. The following proposition allows us to find a measurable Radon-Nikodym derivative with respect to the weak-$\ast$ topology in $L^1(\A_{\Znu})^*$. The proof is technical, so we relegate it to the Appendix, where we also have a slightly more general statement of the result.

\begin{prop}
	Suppose $(\Omega, \mathcal{F}, \P)$ is a probability space and $\A_{\Znu}$ is the quasi-local algebra. Suppose $\Phi$ is a positive normalized linear functional on $L^1(\A_{\Znu})$. Then there is a weakly-$\ast$-measurable RND $G_\omega\in \A_{\Znu}^*$ which is positive and essentially bounded so that
	\begin{equation}\label{eqn:w*RND}
		\Phi(A) = \int_\Omega G_\omega(A(\omega)) d\P\,\quad \forall A\in L^1(\A_{\Znu})\,.
	\end{equation} 
\end{prop} 

We emphasize that this weak-$\ast$ RND is not the same as a Bochner-measurable RND which turns out to rely on the properties of $\A_{\Znu}$ as a Banach space (cf. remarks~\ref{rmk:RND_1} and~\ref{rmk:RND_2} in the appendix below).

\section{Quasilocality Estimates for Disordered Interactions}\label{sec:Quasilocality}
In this section we will discuss the set up for a quantum spin system on the integer lattice $\Znu$ equipped with the $\|\cdot \|_{\infty}$ metric 
\[
d(x,y):= \max_{1\le j \le \nu} |x_j - y_j|\,. 
\] We aim to model local interactions which are disordered but co-vary with respect to a group action and a family of ergodic maps on an underlying probability space. 

The closed ball of radius $r\in \N$ about a point $x\in \Znu$ is the set $b_x(r) = \{y\colon d(x,y) \le r\}$. A simple combinatorial argument shows that $|b_0(r)| = (2r+1)^\nu \le 2^\nu (1+r)^\nu$. Since $d$ is translation invariant and integer valued, we have calculated the size of any closed ball. The upper bound is the so-called $\nu$-regularity property. In fact, $\Z^\nu$ is \emph{$\nu$-hyper-regular} in the sense that there exists a constant $\kappa>0$ so that 
\begin{equation}
	\kappa:= \sup_{r>1} |b_0(r) \setminus b_0(r-1)| (1+r)^{-(\nu-1)} <\infty\,.
\end{equation} 

Recall that an \emph{interaction} is a function from the finite volumes $\h: \mathscr{P}_0(\Znu) \to \aloc$ into the local algebra, so that $\h(Z) = \h(Z)^* \in \A_Z$ for all finite subsets $Z\Subset \Znu$ \cite{BratteliRobinsonII, NachtergaeleOgataSims}. We say a function $\h: \Omega \times \mathscr{P}_0(\Znu) \to \aloc$ is a \textbf{disordered interaction} whenever the following are satisfied: 
\begin{enumerate}[label = \roman*.)]
    \item For all $Z\Subset \Znu$, and almost every $\omega \in \Omega$,  $\h(\omega, Z) = \h(\omega, Z)^* \in \A_Z$. 
    \item The mapping $\omega \mapsto \h(\omega, Z)$ is \emph{weakly locally measurable} in the sense that $\omega \mapsto \varphi(\h(\omega, Z))$ is Borel measurable with respect to $\C$ for all $\varphi \in \A_{\Znu}^*$.
\end{enumerate} 
\begin{rmk}
    Recall that Takeda's Theorem \cite{Takeda} says that any state $\varphi$ (hence any continuous linear functional) on $\A_{\Znu}$ is consistent with respect to volumetric inclusions and also uniquely determined by its finite dimensional restrictions. It is not difficult to show that in the case of a disordered interaction weakly local measurability is equivalent to requiring that for every $Z\Subset \Znu$, the numerical function $\omega \mapsto \varphi(\h(\omega, Z))$ to be Borel measurable for every $\varphi \in \mathcal{A}_{Z}^*$.
\end{rmk}

In this setup, a \textbf{(local) disordered Hamiltonian} in a finite voluem $\Lambda$ is the sum-function  
\[
    \Omega \ni \omega\longmapsto H^{\Lambda}_\omega := \sum_{Z\subset \Lambda} \h(\omega,Z)\in\A_{\Lambda}.
\] Observe that the operator norm may be expressed as a supremum over states
    \[
        \|H^{\Lambda}_\omega \| = \sup_{\phi \in \mathcal{S}(\A_{\Lambda})} |\phi(H^{\Lambda}_{\omega})|\,,
    \] and the right hand side of the quantity is measurable since $\A_{\Lambda}$ is finite dimensional. Given a disordered Hamiltonian with $\P[ \|H^{\Lambda}_\omega \| < \infty]=1$, one can measurably generate a one-parameter unitary group via $H^{\Lambda}_\omega$ in the usual way. Indeed, let 
\[
    U^{\Lambda}_\omega(t) = e^{-itH^{\Lambda}_\omega} = \sum_{n=0}^\infty \frac{(-itH^{\Lambda}_\omega)^n}{n!}\,,
\] which is clearly uniformly convergent for $t$ in compacts $\P$-a.e. $\omega$ since $\|H^\Lambda_\omega\|$ is finite almost surely, and hence $U^{\Lambda}_\omega(t)$ is measurable for all $t\in \R$. The Heisenberg dynamics associated to $H^{\Lambda}_\omega$ are then given by 
\begin{equation}\label{eqn:local_dynamics}
    \alpha^{\Lambda}_{t;\omega}(a) = U^{\Lambda}_\omega(-t)\,a\, U^{\Lambda}_{\omega}(t)\,.
\end{equation} The resulting images are all weakly measurable since the product of measurable linear operators is still measurable \cite[Theorem 2.14]{Bharucha-Reid}.

Recall the $F$-function formalism introduced in \cite{NachtergaeleOgataSims}. Let $F:[0, \infty) \to (0, \infty)$ be a non-increasing function and suppose the following hold: 
\begin{enumerate}
    \item The function $F$ is uniformly summable over $\Znu$ in the sense that 
    \[
        \|F\| := \sup_{x\in \Znu} \sum_{y\in \Znu} F(d(x,y)) <\infty\,,
    \] 
    \item There is a constant $C_F>0$ so that for any $x,y\in \Znu$, one has
    \[
        \sum_{z\in \Znu} F(d(x,z)) F(d(z,y)) \le C_F\, F(d(x,y))\,.
    \]
\end{enumerate} In this case $F$ is called an \emph{$F$-function}. A standard set of examples of $F$-functions is given by the power-law decay $F(x) = (1+x)^{-\nu - (1+\epsilon)}$ for a small $\epsilon >0$ on a $\nu$-regular discrete metric space (see \cite{NachtergaeleOgataSims}). Given a disordered interaction $\h$, one can define a $[0, \infty]$-valued random variable 
\begin{equation}\label{eqn:disordered_F_norm}
    \calN(\omega) = \sup_{x,y\in \Znu}  \sum_{\substack{Z\Subset\Znu: \\ Z\ni x,y}}\,\|\h(\omega, Z)\|F(d(x,y))^{-1}\,,
\end{equation} where we allow the possibility of infinite values for now (later we will restrict to the case that $\calN$ is finite almost surely). Note that $\calN$ is just the $F$-norm of $\h$ on at a given realization of the disorder $\omega$ \ \textemdash \ if the interaction $h$ is ergodic, as in \eqref{eqn:translation-covariance} or \eqref{eqn:ergodic_interaction} below, it follows that $\calN$ is deterministic.  Whenever $\calN(\omega)$ is finite, equation~(\ref{eqn:disordered_F_norm}) implies that for every finite $Z\Subset \Znu$, the norm-function $\|\h(\omega,Z)\|$ is bounded: 
\[
    \|\h(\omega, Z)\| \le \calN(\omega) \sum_{x,y\in Z} F(d(x,y)) \ .
\] Thus the following rough estimate holds almost surely for any disordered interaction
    \begin{equation}\label{eqn:calN_to_esssup}
        \|H^{\Lambda}_\omega\| \le \calN \sum_{x,y\in \Lambda} \sum_{\substack{Z\subset \Lambda \colon\\ x,y\in Z}}F(d(x,y)) = 2^{|\Lambda|-2}\calN(\omega) \sum_{x,y\in \Lambda}F(d(x,y))\,.
    \end{equation}

\begin{prop}[Lieb-Robinson Bound]\label{prop:disordered_LRB}
    Let $\h$ be a disordered interaction with $\P[\calN<\infty]=1$. Let $X,Y\Subset \Znu$ be disjoint sets and let $a\in \A_{X}$ and $b\in \A_{Y}$. Then, whenever $\Lambda \Subset \Znu$ satisfies $X\cup Y\subset \Lambda$ one has 
    \begin{equation}
        \|[b, \alpha^{\Lambda}_{t;\omega}(a)]\| \le \frac{\|a\| \|b\| }{C_F} (e^{\calN C_F |t| }-1 )\sum_{x\in X} \sum_{y\in Y} F(d(x,y))\,,
    \end{equation} on a set with probability one.
\end{prop}
 The proof of this result is a simple `insert-$\omega$' corollary to the well-established proofs of quasilocality estimates in (e.g.) \cite[Theorem 1]{NachtergaeleOgataSims} or \cite[Theorem 2]{NachtergaeleVershyninaZagrebnov}, etc., so we will not repeat it here. For state-of the art estimates see Theorem 3.1 in \cite{NachtergaeleSimsYoung}. We now turn to the existence of the thermodynamic limit. One arrives at the following corollary of the almost-sure Lieb-Robinson bound.
\begin{cor}\label{cor:Thermo_AS}
    Let $\h$ be a disordered interaction with $\calN<\infty$ almost surely. Let $H^{\Lambda}_\omega$ be the disordered local Hamiltonian and let also $\alpha^{\Lambda}_{t;\omega}$ be the associated one-parameter group of $*$-automorphisms. Let $\Lambda_n$ be an IAS. Then for all $\omega \in [\calN<\infty]$ the following limit exists 
        \begin{equation}
            \lim_{n\to \infty} \alpha^{\Lambda_n}_{t;\omega} (a_X) =: \alpha^{\Znu}_{t;\omega}(a_X)\,,
        \end{equation} for all local observables $a_X\in \aloc$. The propagator dynamics $\alpha^{\Znu}_{t,\omega}$ is weakly-$\ast$ measurable in the sense that for all $\varphi \in \A_{\Znu}^*$ and all $a\in \aloc$, the function 
        \[
            \Omega \ni \omega \mapsto \varphi(\alpha^{\Znu}_{t;\omega}(a_X))\in \C\,,
        \] is measurable. Moreover, $\alpha_{t;\omega}$ is independent of the IAS $\{\Lambda_n\}_{n=1}^\infty$. 
\end{cor}
\begin{proof}
    Following the methods of \cite{NachtergaeleOgataSims, NachtergaeleSimsYoung}, we find that for any $a_X\in \aloc$, if $n>m\ge N$ where $N\gg0$ is large enough so that $\Lambda_N\supset X$, the following estimate holds $\omega$-point-wise by Duhamel's principle.
    \[
        \|\alpha^{\Lambda_n}_{t;\omega}(a_X) - \alpha^{\Lambda_m}_{t;\omega}(a_X)\| \le \sum_{\substack{Z\subset \Lambda_n \colon\\ Z\not \subset \Lambda_m}}  \int_0^{|t|} \|[\h(\omega, Z), \alpha^{\Lambda_n}_{s;\omega}(a_X)]\| ds \le \sum_{p\in \Lambda_n \setminus \Lambda_m} \sum_{\substack{Z\subset \Lambda_n\colon \\ p \in Z}}\int_0^{|t|} \|[\h(\omega, Z), \alpha^{\Lambda_n}_{s;\omega}(a_X)]\| ds\,.
    \] Applying Proposition~\ref{prop:disordered_LRB} $\omega$-pointwise, we obtain an upper bound of the form 
    \[
         \|\alpha^{\Lambda_n}_{t;\omega}(a_X) - \alpha^{\Lambda_m}_{t;\omega}(a_X)\| \le  \frac{\|a_X\|}{C_F} \left(\int_0^{|t|} e^{\calN |s|}  ds \right) \sum_{p\in \Lambda_n\setminus \Lambda_m} \sum_{\substack{Z\subset \Lambda_n \colon\\ p\in Z}}  \|\h(\omega, Z)\| \cdot \sum_{x\in X}\sum_{y\in Z} F(d(x,y))
    \] for all $\omega$ in the event $[\calN<\infty]$. From this, use Fubini's theorem to rearrange the sums to obtain 
    \[
         \|\alpha^{\Lambda_n}_{t;\omega}(a_X) - \alpha^{\Lambda_m}_{t;\omega}(a_X)\| \le  \|a_X\| \calN(\omega) \left(\int_0^{|t|} e^{\calN |s|}  ds \right) \sum_{x\in X}\sum_{p\in \Lambda_n\setminus \Lambda_m} F(d(x,p))\,,
    \] which tends to zero as $n,m \to \infty$ almost surely because $[\calN<\infty]$ with probability one. Furthermore, we see that $\alpha^{\Znu}_{t;\omega}(a_X)$ is the limit of a sequence of weakly measurable functions, namely $\alpha^{\Lambda_n}_{t;\omega}(a_X)$, so it is measurable for all $a_X$. But this is just a re-phrasing of the weakly-$\ast$ measurable requirement. 
 \end{proof}

\begin{note}
    To keep the exposition tidier, we will often write $\alpha_{t;\omega}$ instead of $\alpha^{\Znu}_{t;\omega}$, reserving the latter only when we need to emphasize the thermodynamics over finite-volume dynamics. 
\end{note}

\section{Main Results}
This section is divided into three main parts. The first subsection introduces our notion of Ergodic Local Hamiltonian. As noted in the introduction, a similar definition appears as far back as \cite{PasturFigotin}. Our focus however is significantly more general and we establish ergodic co-variance equations for the propagator dynamics in the thermodynamic limit. Then, using Lemma~\ref{lem:hitting_time_compact}, we establish the existence of measurable infinite-volume ground states for the thermodynamics. In the second subsection, we show that among the measurable bulk ground state sectors, one can construct a sector that co-varies along with the translation automorphisms. The key tool in that proof is Proposition~\ref{prop:w*_simple_functions}. We conclude by proving our second main result, Theorem~\ref{thmx:spec} that in the GNS representation with respect to the ergodic ground state, the spectrum of the bulk Hamiltonian is constant. 

\subsection{Ergodic Local Hamiltonians}
The key theorem to our approach is the existence of an ergodic ground state for the family of Hamiltonians $\{H^{\Lambda}_{\omega}: \Lambda \Subset \Znu\}$ in the thermodynamic limit. Once this is established, the proof that the bulk excitation spectrum is deterministic is straightforward. Let us begin by restricting our attention to a class of disordered interactions.

\begin{define}
    Given a disordered interaction $\h:\Omega \times \mathscr{P}_0(\Znu) \to \aloc$, we say that $\h$ is a $\Znu$-\textbf{ergodic interaction} if there exists a group action of $\Znu$ on $\Omega$ by measure-preserving, ergodic maps $\{\vartheta_x\colon\Omega \to \Omega\}_{ x\in \Znu}$  so that
\begin{equation}\label{eqn:ergodic_interaction}
    \tau_x(\h(\omega, Z)) = \h(\vartheta_x \omega, \tau_x(Z))\,,
\end{equation} where $\tau_x(Z) = \{z+x\colon z\in Z\}$. By a slight abuse of terminology, we will drop the $\Znu$ adjective and just write ``\emph{ergodic interaction}". The next lemma is a result of a straightforward computation.
\end{define}

\begin{lem}\label{lem:local_covariance}
    Let $H^{\Lambda}_\omega$ be the local Hamiltonian associated to an ergodic interaction $\h$. Then, the following relations hold for any $\Lambda \Subset \Znu$:
        \begin{equation}
            H^{\Lambda}_{\omega} = \tau_{-x}(H^{\tau_x(\Lambda)}_{\vartheta_x \omega})\,,
        \end{equation} and moreover, 
        \begin{equation}
            \alpha^{\Lambda}_{t;\omega} = \tau_{-x} \circ \alpha^{\tau_{x}(\Lambda)}_{t;\vartheta_x \omega} \circ \tau_x. 
        \end{equation} 
\end{lem}

\begin{cor}\label{cor:N-norm_const,thermo_covariance}
    Let $\h:\Omega \times \mathscr{P}_0(\Znu)$ be an ergodic interaction and consider the family of local dynamics $\{\alpha^{\Lambda}_{t;\omega}\}_{\Lambda \Subset \Znu}$ associated to $\h$. Suppose that $\P[\calN<\infty]>0$. Then the following hold. 
    \begin{enumerate}
        \item The quantity $\calN$ is almost surely constant. 
        \item The thermodynamic limit exists and satisfies the following ergodicity condition with probability one.  
                \begin{equation}
                    \alpha_{t;\omega} = \tau_{-x} \circ \alpha _{t;\vartheta_x \omega} \circ \tau_x\, \quad \forall x\in \Znu. 
        \end{equation}  
    \end{enumerate}

\end{cor}
\begin{proof}
To prove 1., our aim is to show that $\calN(\omega) = \calN(\vartheta\omega)$. Observe that for any pair of points $x,y\in \Znu$, the quantities 
\begin{equation}
    N_{x,y,n}(\omega) := \sum_{\substack{Z\ni x,y\\ \diam(Z)\le n}} \|\h(\omega, Z)\| F(d(x,y))^{-1}
\end{equation} form an increasing sequence of positive random variables which satisfy $\sup_{x,y}\lim_{n\to \infty} \mathcal{N}_{x,y,n} = \calN$. 
For any $z \in \Znu$, we have
\begin{align*}
    N_{x,y,n}(\vartheta_z \omega) &= \sum_{\substack{Z\ni x,y\\ \diam(Z) \le n}} \|\h(\vartheta_z \omega, Z) \| F(d(x,y))^{-1} = \sum_{\substack{Z\ni x,y\\ \diam(Z) \le n}} \|\tau_z \h(\omega, Z-z ) \| F(d(x-z,y-z))^{-1} \\
        &= \sum_{\substack{W\ni x-z, y-z\\ \diam(W)\le n}} \|\h(\omega, W)\| F(d(x-z, y-z))^{-1} = N_{x - z, y-z, n}(\omega)\,.
\end{align*} Note we have crucially used translation invariance of the $\ell^\infty$-metric $d$ on $\Znu$ to obtain the second equality. Now, taking the supremum over $x,y\in \Znu$ and $n\in \N$, we obtain the equation $\calN(\omega) = \calN(\vartheta_z \omega)$. Whence $\vartheta_z ^{-1}[\calN < \infty] \subset [\calN<\infty]$. 

For 2., the estimate~(\ref{eqn:calN_to_esssup}) shows that the local dynamics $\{\alpha^{\Lambda}_{t;\omega}: \Lambda \Subset \Znu\}$ are essentially bounded almost surely. Furthermore, by combining part 1 and  Proposition~\ref{cor:Thermo_AS}, the strong limit 
\[
    s-\lim_{n\to \infty} \alpha^{\Lambda_n}_{t;\omega} = \alpha_{t;\omega}\,,
\] exists with probability one for every IAS of finite volumes $\Lambda_1 \subset \Lambda_2 \subset \cdots$ with $\bigcup \Lambda_n = \Znu$. Now, suppose $a\in \A_{X}$ where $X\Subset \Znu$. For any $\epsilon>0$, on a set of probability one, there is some $n$ large enough so that \[ \max\{
\|\alpha_{t;\omega}(a) - \alpha^{\Lambda_n}_{t;\omega} \|, \|\alpha_{t; \vartheta_x \omega} \circ \tau_x(a) - \alpha_{t; \vartheta_x \omega}^{\Lambda_n - x}\circ \tau_x(a)\|\}  < \epsilon/2\,,\] since $\Lambda_n - x$ is still an IAS.  By Lemma~\ref{lem:local_covariance}, we now have \begin{align*}
    \|\alpha_{t;\omega}(a) - \tau_{x}^{-1} \circ \alpha_{t; \vartheta_x \omega} \circ \tau_x(a)\| & \le \|\alpha_{t;\omega}(a) - \alpha^{\Lambda_n}_{t;\omega}(a) \| + \|\alpha^{\Lambda_n}_{t;\omega}(a) - \tau_{x}^{-1} \circ \alpha_{t; \vartheta_x \omega} \circ \tau_x(a)\|\\
    & \le \frac{\epsilon}{2} + \|\alpha^{\Lambda_n + x}_{t;\vartheta_x \omega}\circ \tau_x(a) - \alpha_{t;\vartheta_x \omega} \circ \tau_x (a)\| <\epsilon\,.
\end{align*} Therefore $\alpha_{t;\omega}(a) = \tau_x^{-1} \circ \alpha_{t;\vartheta_x \omega} \circ \tau_x(a)$ for any $a\in \aloc$, which is a norm-dense $*$-subalgebra of $\A_{\Znu}$. 
\end{proof}

\subsection{Covariant Disordered Ground States}
The next step in our analysis is to establish the existence of a weakly-$\ast$-measurable state $\Omega \ni \omega \mapsto \gamma_\omega \in \mathcal{S}(\A_{\Znu})$ that satisfies both the ground state property for all $a\in \aloc$, and the translation co-variance property $\gamma_\omega \circ \tau_{x} = \gamma_{\vartheta_x \omega}$.
Suppose we have some ergodic local interaction $\h$. Define the unbounded disordered derivation $\delta_\omega$ on $\A_\Znu$ via 
\begin{equation}\label{def:unbounded_derivation}
	\delta_\omega(a_X) = \sum_{Z\cap X \neq \emptyset} [i\h(\omega,Z), a_X]\,,
\end{equation} for all $a_X\in \aloc$. Note that $\dom(\delta_\omega) \supset \aloc$ for every $\omega \in \Omega$, so $\delta_\omega$ is densely defined. We say that a weakly-$\ast$-measurable function $\Gamma: \Omega \to \mathcal{S}(\A_{\Znu})$ is a \textbf{disordered ground state} if the ground state property holds almost surely:
\begin{equation}\label{eqn:C_star_ground_state}
	-i\Gamma_\omega(a^*\delta_\omega(a)) \ge 0, \quad \forall a\in \aloc,\, \P-\mathrm{a.e. } \omega\,.
\end{equation}

\begin{lem}\label{lem:gs}
	Let $(\Omega, \mathcal{F}, \P)$ be a probability space and $\A_{\Znu}$ be the quasilocal algebra. Suppose $\h$ is an ergodic interaction and $\omega\in \Omega$ such that  $\calN(\omega)<\infty$. Let $(\Lambda_n \colon n\in \N)$ be an IAS. If for each $n$, $\lambda^{(n)}_\omega$ is a ground state for the local dynamics $\alpha^{\Lambda_n}_{t;\omega}$ generated by $\h$, then any weak-$\ast$ limit of the sequence $\lambda^{(n)}_\omega$ is a ground state of $\alpha^{\Znu}_{t;\omega}$. 
\end{lem}
\begin{proof}
An application of Lemma~\ref{lem:hitting_time_compact} finds a random increasing sequence $\sigma_n \to \infty$ and a measurable weak-$\ast$ limit point $\Gamma_{\omega}\in\A_{\Znu}^*$. It is straightforward to check that $\Gamma_\omega$ is a state $\P$-a.e. $\omega$.  We need to verify that $\Gamma_{\omega}$ satisfies the relevant positivity condition. 

Write $\delta_n = \sum_{Z\subset \Lambda_n} [i\h(Z,\omega), \,\cdot \,].$ By Corollary~\ref{cor:Thermo_AS}, the strong limit $\alpha^{\Lambda_n} \to \alpha^{\Znu}$ exists with probability one. Let $\delta_\omega$ denote the unbounded derivation defined as in~(\ref{def:unbounded_derivation}). 
    For each $\omega\in [\calN <\infty]$, the derivation $\delta_\omega$ is closable with dense domain containing $\A^{\loc}_{\Znu}$ as a core, and $\alpha_{t;\omega} = \exp\{t \overline{\delta_\omega}\}$ (see Bratteli and Robinson \cite[Proposition 3.2.22, Example 3.2.25]{BratteliRobinsonI}). Therefore, for any $a\in \dom(\overline{\delta_\omega})$, we can find a sequence  $a_n\in \A_{\Lambda_n}$ (which may depend on $\omega$ but need not be measurable in $\omega$) for which $a_n \to a$ and $\delta_{n;\omega}(a_n) \to \delta_{\omega}(a)$.  We observe that along the evaluations of the random sequence $\sigma_n(\omega)$ the corresponding disordered ground state condition holds for all $\omega$ in a probability one set: 
    \begin{equation}\label{eqn:pf:ARS_local}
         -i\lambda^{(\sigma_n)}_{\omega}(a^*_{\sigma_n}\delta_{\sigma_n}(a_{\sigma_n})) \ge0\,.
    \end{equation}Therefore, taking the limit $n\to \infty$, we arrive at 
    \[
        -i\Gamma_\omega(a^*\delta_\omega(a)) \ge 0 \text{ almost surely}.
    \] Hence $\Gamma$ is a disordered ground state.
\end{proof}

\begin{rmk}\label{rmk:Gibbs_1}
	The argument we give can be easily generalized to show that weak-$\ast$ limits of local Gibbs states satisfy a disordered version of the $(\alpha, \beta)$-KMS condition in the thermodynamic limit. This is essentially a disordered version of Propositions 3.5.25 and 6.2.15 of \cite{BratteliRobinsonII} if one replaces~(\ref{eqn:pf:ARS_local}) with the relevant Araki-Roepsdorff-Sewell conditions (see \cite[Theorem 3.5.15]{BratteliRobinsonII}). 
\end{rmk}

Now that we have established that disordered ground states exist generally, we want to show that there are disordered ground states satisfying the same ergodicity condition as~(\ref{eqn:ergodic_interaction}).
\begin{thm}[Theorem~\ref{thmx:exist}]\label{thm:exist}
	Let $(\Omega, \mathcal{F}, \P)$ be a countably generated probability space. Let $\h$ be an ergodic interaction with $\P[\calN <\infty]>0$ for some deterministic $F$-function, $F$.  There exists a weak-$\ast$ measurable function $\Gamma$ that is a disordered ground state of $\h$ satisfying 
	\begin{equation}
		\Gamma_\omega \circ \tau_x = \Gamma_{\vartheta_x \omega}\quad \forall x\in \Znu \text{ almost surely.}
	\end{equation} 
\end{thm}
\begin{proof}
	Let $\gamma_\omega$ be a disordered ground state for $\h_\omega$ which exists by Lemma~\ref{lem:gs}. Notice that $\gamma_{\vartheta_p\omega}$ is a ground state for $\h_{\vartheta_p\omega}$ for any $p\in \Znu$, but by the covariance of the thermodynamic limit, this means $\gamma_{\vartheta_p\omega}$ is a ground state for $\tau_p(\h_\omega)$. With this in mind, we define the following sequence of `averaged' ground states of $\h_\omega$
	\begin{equation}
		\gamma^{(n)}_\omega = \frac{1}{|b_0(n)|} \sum_{p\in b_0(n)} \gamma_{\vartheta_p \omega} \circ \tau_{-p}\,.
	\end{equation} Notice that $\|\gamma^{(n)}_\omega\|_{\A_{\Znu}^*} \equiv 1$ by the Russo-Dye theorem. In view of Lemma~\ref{lem:pairing} we can regard the $\gamma^{(n)}_\omega$ as elements of the dual space $[L^1(\A_{\Znu})]^*$. But, since $(\Omega, \mathcal{F}, \P)$ is countably generated, $L^1(\A_{\Znu})$ is separable by \cite[Proposition 1.2.29]{Hytonen_et_al_vol1}. Thus by general functional analysis the unit ball $[L^1(\A_{\Znu})]^*$ is weakly-$\ast$ sequentially compact. Thus, we lose no generality by assuming $\gamma^{(n)}$ converges to, say, $\Phi$ in the $\sigma(L^1(\A_{\Znu}), L^1(\A_{\Znu})^*)$ topology. By Proposition~\ref{prop:w*RND}, we conclude there is some $\Gamma:\Omega \to \A_{\Znu}^*$ which is weakly-$\ast$ measurable representing $\Phi$ whose norm function is essentially bounded. Finally, by the continuity of scalar-valued Radon-Nikodym derivatives \cite{Bermudezet_al} for the associated complex measures $\nu_{\Gamma}(\, \cdot; a)$ and $\nu_{\gamma^{(n)}}(\,\cdot;a)$ as in~(\ref{eqn:induced_measure}), we conclude $\Gamma_\omega(a) = \lim_{n\to \infty} \gamma^{(n)}_\omega(a)$ for almost every $\omega$ and every $a\in \A_{\Znu}$, hence $\Gamma_\omega$ is a state on $\A_{\Znu}$ almost surely.  
	
	Let $z$ be a unit length generator for $\Znu$. By $\nu$-hyperregularity, the symmetric difference of $b_0(n)$ and $b_z(n)$ obeys $|b_0(n)\Delta b_z(n)| \lesssim n^{\nu-1} $ and $|b_0(n)| \sim n^\nu$. Now, 
	\begin{align*}
		\left| \gamma^{(n)}_{\vartheta_z\omega}(a) - \gamma^{(n)}_\omega \circ \tau_z\right| & = \frac{1}{|b_0(n)|} \left|\sum_{p\in b_0(n)} \gamma_{\vartheta_{p+z}\omega} \circ \tau_{-p}(a) - \sum_{p'\in b_0(n)} \gamma_{\vartheta_p \omega} \circ \tau_{z-p}(a) \right|\\
		&= \frac{1}{|b_0(n)|} \left| \left(\sum_{p\in b_0(n) \setminus b_z(n)} + \sum_{p\in b_z(n) \setminus b_0(n)} \right)\gamma_{\vartheta_p \omega}\circ \tau_{-p}\right|\\
		&\le \|a\| |\gamma_\omega(\one)| \frac{|b_0(n) \Delta b_z(n)|}{|b_0(n)|} = \mathcal{O}(n^{-1}) \text{ as }n\to \infty\,,
	\end{align*} where the second equality follows by re-indexing $p' = z-p$ and canceling the differences. In particular, this shows $\esssup_{\omega\in \Omega} \left| \gamma^{(n)}_{\vartheta_z\omega}(a) - \gamma^{(n)}_\omega \circ \tau_z\right| = \mathcal{O}(n^{-1})$ as $n\to \infty$.
	
	Therefore, if $\Gamma$ is the weak-$\ast$ limit (with respect to the weak-$\ast$ topology on the Bochner space $(L^1(\A))^*$) of the $\gamma^{(n)}$'s, an $\epsilon/3$-argument combined with the estimate above shows that with probability one,
	\begin{equation}
		\Gamma_{\vartheta_z \omega}(a) = \Gamma_\omega\circ \tau_z(a)\,,
	\end{equation} for all $a\in \aloc$, almost surely. But $z$ was an arbitrary generator of $\Znu$ and $\aloc$ is norm-dense in $\A_{\Znu}$, so we conclude that $\Gamma$ is ergodic. 
\end{proof}

\begin{rmk}
    Similar to Remark~\ref{rmk:Gibbs_1} above, the conclusions of Theorem~\ref{thm:exist} can be restated in terms of w*-limits of disordered local Gibbs states, to show there is a covariant $(\alpha, \beta)$-KMS state. 
\end{rmk}

\subsection{The Spectrum of a Bulk Hamiltonian in an Ergodic Ground State}
Since we have established Theorem~\ref{thm:exist}, we may now prove our second main result, Theorem~\ref{thmx:spec}. We shall state this as Corollary~\ref{thm:Bulk_gap_const} to a much more general assertion Theorem~\ref{thm:spatial_covariance} which can be proved without reference to a specific physical model. 

\begin{thm}\label{thm:spatial_covariance}
    Let $\A$ be a separable unital $C^*$-algebra, and let $(\Omega, \mathcal{F}, \P)$ be a probability space. Suppose $G$ is a discrete group acting on $\A$ via $\tau$, and suppose $G$ also acts on  $(\Omega, \mathcal{F}, \P)$ via family of measure preserving ergodic maps $\{\vartheta_g : g\in G\}$. Suppose $\{S_{t;\omega}\colon \A\to \A\}_{\omega \in \Omega}$ is a family of strongly continuous one-parameter groups of $*$-automorphisms which are weakly-$\ast$ measurable in $\omega$. Further, suppose that $S_{t; \omega}$ obey the following co-variance condition 
        \begin{equation}
            S_{t;\omega} = \tau_g^{-1} \circ S_{t;\vartheta_g\omega} \circ \tau_g; \quad \forall g\in G,\, \P-a.e. \omega\,.
        \end{equation} Assume $S_{t;\omega}$ admits a ground state $\psi_\omega$ so that $\psi_\omega\circ S_{t;\omega} = \psi_\omega$ which also satisfies the following covariance condition
        \begin{equation}
            \psi_\omega = \psi_{\vartheta_g \omega}\circ \tau_g\quad  \text{$\P$ almost surely}\,.
        \end{equation} 
        
        If $(\mathcal{H}_\omega, \pi_\omega, \Psi_\omega)$ is the GNS triple associated to a covariant state $\psi_\omega$ on $\A$, then the following hold
        \begin{enumerate}
            \item For $\P$-a.e. $\omega$, there is a densely defined self-adjoint operator $H_\omega$ in the GNS space $(\mathcal{H}_\omega, \pi_\omega, \Psi_\omega)$ which generates the implemented dynamics such that $\inf \sigma(H_\omega) = 0$ given by 
            \[
                S_{t;\omega}(\pi_\omega(a)\Psi_\omega) := \pi_\omega(S_{t;\omega}(a)) \Psi_\omega \quad \forall t\in \mathbb{R}, a\in \A_{\Znu}\,.
            \]
            \item For $\P$-a.e. $\omega$, there exist unitaries $U_g: \mathcal{H}_\omega \to \mathcal{H}_{\vartheta_g\omega}$ which intertwine $H_\omega$ in the following way
            \begin{equation}
                H_\omega = U_g^* H_{\vartheta_g \omega} U_g\,.
            \end{equation}
            \item There is an event $\Omega_0$ which has full probability for which 
    \begin{equation}
        \spec_{*}(H_\omega) \text{ is deterministic for }*\in \{\mathrm{ess}, \mathrm{dis}, \mathrm{c}, \mathrm{ac}, \mathrm{pp}, \mathrm{sc}\}\,.
    \end{equation}
        \end{enumerate} 
\end{thm}
\begin{proof}
    Recall that the GNS Hilbert space $\mathcal{H}_\omega$ is formed by completion of the quotient space $A/J_\omega$ with respect to the inner product $\<\,\cdot\,|\,\cdot\,\>_\omega$ induced by $\psi_\omega$, where $J_\omega := \{a\in A\colon \psi_\omega(a^*a) =0\}$. We denote by $[a,\omega]$ the equivalence class of $a\in A$ modulo $J_\omega$. In particular, the cyclic vector $\Psi_\omega = [\one, \omega]$. Since $\A$ is separable, the resulting Hilbert spaces $\mathcal{H}_\omega$ are also separable. In particular, for every $\omega$, we can pick on orthonormal basis of $\mathcal{H}_\omega$ consisting of vectors of the form $\{[a,\omega]\colon a\in \aloc\}$, by the Gram-Schmidt process.

    We define the $U_g$ in the following way. Let $U_g: \pi_\omega(A)\Psi_\omega \to \mathcal{H}_{\vartheta_g \omega}$ via
    \begin{equation}
        U_g[a,\omega] = [\tau_g(a), \vartheta_g \omega]\,.
    \end{equation} We note that 
    \[
        \<U_g [a, \omega]| U_g[b,\omega]\>_{\vartheta_g \omega}= \psi_{\vartheta_g\omega} ( \tau_g(a^*) \tau_g(b)) = \psi_{\vartheta_g \omega} \circ \tau_g (a^*b) = \<[a, \omega]| [b, \omega]\>_\omega\,,
    \] by the co-variance property. Whence $U_g$ is well defined for all $g\in G$ and moreover the $U_g$ are densely defined linear operators which preserve the inner product, so they extend uniquely to unitaries $U_g : \mathcal{H}_\omega \to \mathcal{H}_{\vartheta_g \omega}$. 

    By standard theory we know that $S_{t;\omega}$ can be implemented by a strongly continuous one parameter group (of endomorphisms) on $\mathcal{H}_\omega$. By abuse of notation, let us identify $S_{t;\omega}$ with $S_{t;\omega}[a,\omega] := [S_{t;\omega}(a), \omega]$. Then, $S_{t;\omega}$ is generated by a densely defined self-adjoint operator $H_\omega$. By the covariance condition, we see that 
        \begin{align*}
            U_g^* H_{\vartheta_g \omega} U_g [a, \omega] &= [ \tau_g\lim_{t\downarrow 0}\frac{S_{t;\vartheta_g \omega}(\tau_g(a)) - \tau_g(a)}{t}, \omega]= [\lim_{t\downarrow 0} \frac{S_{t;\omega}a -a}{t}, \omega]= H_\omega [a, \omega]\,,
        \end{align*} as required. 

By the intertwining relation, any spectral projection $P^{(\lambda, \mu)}_\omega := \chi_{(\lambda, \mu)}(H_\omega)$ where $\lambda < \mu$ are rational, satisfies $U_g^*P^{(\lambda, \mu)}_\omega U_g= P^{(\lambda,\mu)}_{\vartheta_g \omega}$. Furthermore, $P^{(\lambda, \mu)}_\omega$ has a measurable trace-function. Indeed, given any orthonormal basis $\{[a_n, \omega]\colon a_n\in \aloc\}_{n\in \N}$ for $\mathcal{H}_\omega$, we find 
    \[
        \Tr[P^{(\lambda, \mu)}_\omega] = \sum_{n=1}^\infty \<[a_n,\omega]| P^{(\lambda,\mu)}_\omega |[a_n,\omega]\rangle_{\mathcal{H}_\omega} = \sum_{n=1}^\infty \< [\tau_g(a_n), \vartheta_g\omega]| P^{(\lambda, \mu)}_{\vartheta_g \omega} |[\tau_g(a_n), \vartheta_g\omega]\>_{\mathcal{H}_{\vartheta_g \omega}} = \Tr[P^{(\lambda,\mu)}_{\vartheta_g\omega}]\,.
    \]In particular, since $\vartheta_g$ are ergodic, this means that the trace of all the spectral projections of $H_\omega$ are constant for fixed $\lambda, \mu$. The rest of the proof now follows the same argument as \cite{KirschMartinelli}. Namely, the spectral family of $H_\omega$ has dimension jumps which are independent of $\omega$. Since this characterizes points in the spectrum (cf. \cite[Theorems 7.22 and 7.24]{Weidmann}) we are done. 
\end{proof}

One immediately has the following. 
\begin{cor}[Theorem~\ref{thmx:spec}]\label{thm:Bulk_gap_const}
    Let $\h$ be an ergodic interaction with $\calN<\infty$ almost surely. Suppose $\varphi_\omega$ is a covariant ground state of the thermodynamics $\alpha_{t;\omega}$ on $\A_{\Znu}$. Let $(\mathcal{H}_\omega, \pi_\omega, \Phi_\omega)$ denote the GNS triple associated to $\varphi_\omega$ and let $H_\omega$ denote the GNS Hamiltonian associated to the strongly continuous group of unitaries $U_{t;\omega}$ which implement $\alpha^{\Znu}_{t;\omega}$ in $\mathcal{H}_\omega$. Then, $\spec(H_\omega)$ is deterministic.
\end{cor}

\begin{rmk}
    This pair of results is strongly reminiscent of the classic theorem of Pastur \cite{Pastur75} later generalized by Kirsch and Martinelli \cite[Theorem 1]{KirschMartinelli}, but is distinct in several key ways. Besides being a result for quantum spin systems, one key difference is that in the proof of Theorem~\ref{thm:spatial_covariance} the underlying Hilbert space that the unitaries map onto is $\omega$-dependent. We are able to circumvent this by using separability of the fibers, and independence of the canonical trace from a choice of orthonormal basis. A version of the Pastur-Kirsh-Martinelli theorem with $\omega$-dependent fibers is given in \cite{Lentz_et_al} in the context of measurable groupoids. 
\end{rmk}
\subsection{An Application}
In this section we make some contact with the disordered spin chain models we mentioned in the Introduction by considering a random perturbation of an interaction in $\Znu$. Proposition~\ref{prop:apply} below shows that for interactions of a specific form resembling the disordered XY and XXZ models has a deterministic spectrum in its ground state. The key innovation is that our proof works for arbitrary dimension $\nu$, whereas the results for the XY and XXZ results depend on the Jordan-Wigner theorem which only applies if $\nu = 1$. 
\begin{prop}\label{prop:apply}
    Let $\Phi$ be a deterministic, translation invariant interaction on $\Z^\nu$ with $\|\Phi\|_F<\infty$ for some $F$-function $F:[0, \infty) \to (0,\infty)$. Suppose $(\lambda_x)_{x\in \Znu}$ is an array of IID random variables on a common probability space $(\Omega, \mathcal{F}, \P)$ which are almost surely finite. Suppose $v\in \A_{\{0\}}$ and define 
    \begin{equation}
        \Psi(X, \omega) = \Phi(X) + \delta_{1}(|X|)\sum_{x\in X} \lambda_x(\omega) \tau_x(v) \quad \forall X\Subset \Znu
    \end{equation} where $\delta$ is the Kronecker-delta. Notice that the total Hamiltonian in a finite volume is then given by 
    \[
        H^\Lambda_\omega = \sum_{Z\subset \Lambda}\Psi(Z,\omega) = \sum_{Z\subset \Lambda}\Phi(Z) + \sum_{x\in \Lambda} \lambda_x(\omega) v_x\,.
    \] Then, $\Psi$ defines an ergodic interaction. In particular, the following hold:
    \begin{enumerate}
        \item There is an ergodic ground state $\psi_\omega$ for $\Psi$. 
        \item With respect to any ergodic ground state, the spectrum of the GNS Hamiltonian $H$ is deterministic. 
    \end{enumerate} Furhtermore, in the case that the $\lambda_x$ are essentially bounded, we can give a quantitative estimate on the $F$-norm of the perturbation via \[
        \|\Psi\|_F = \|\Phi\|_F + F(0)^{-1} \|v\|\, \|\lambda\|_{L^\infty(\P)}\,.
        \]
\end{prop}
\begin{proof}

It is a standard application of Kolmogorov's Extension Theorem to see that we can replace the original probability space in such a way that  the IID $\lambda_x$ are recovered as ergodic shifts of of a single continuous random variable with the same law as $\lambda_0$. 

   Let $\alpha^{\Lambda,\Psi}_{t;\omega}$ denote the Heisenberg dynamics generated by $\Psi_\omega$ and let $\alpha^{\Lambda, \Phi}_{t}$ denote the unperturbed Heisenberg dynamics. If the $\lambda_x$ are essentially bounded, it is a simple calculation to show that 
    \[
        \|\Psi\|_F \le \|\Phi\|_F + F(0)^{-1}\|v\| \|\lambda\|_{L^\infty(\P)},
    \] so we may apply Proposition~\ref{cor:Thermo_AS} to conclude that the thermodynamic limit of the $\alpha^{\Lambda, \Psi}_{t;\omega}$ exists almost surely. 
    
    We claim the thermodynamic limit still holds if the $\lambda_x$ are only finite almost surely. Postponing the proof of this claim for now, we can obtain Proposition~\ref{prop:apply} by arguing in the same way as the proof of Lemma~\ref{lem:gs} to obtain a global ground state, and then as in the proof of Theorem~\ref{thm:spatial_covariance} to obtain an ergodic disordered ground state. Now the result follows by straightforward application of Corollary~\ref{thm:Bulk_gap_const}.

    To prove the claim, consider the family of random local unitaries  $u_x(t) := e^{it\lambda_x v_x}$ which are mutually commuting for any $x\neq y$. Defining $T^{\Lambda}_\omega(t) = \bigotimes_{x\in \Lambda} u_x(t)$ we observe that $T^{\Lambda}(t) = \exp\{it\sum_{x\in \Lambda} \lambda_x v_x\}$. Consider the transformation on $\alpha^{\Lambda,\Psi}_{t;\omega}$ given by 
    \[
        \tilde \alpha_t^{\Lambda}(\, \cdot\,) := (T^{\Lambda}_\omega (t))^*\alpha^{\Lambda, \Psi}_{t;\omega}(\, \cdot\,)T^{\Lambda}_\omega(t)\,.
    \] Let $(\Lambda_n)_{n=1}^\infty$ be an IAS. By the product structure of the unitary transformations $T^{\Lambda}_\omega$, we observe that for any $n\ge m$, and any finite volume $B\supset \Lambda_n$, one has 
    \begin{align*}
        \|\alpha^{\Lambda_n, \Psi}_{t;\omega}(a_X) - \alpha^{\Lambda_m, \Psi}_{t;\omega}(a_X)\| &=\|T^B(t)^*( \alpha^{\Lambda_n, \Psi}_{t;\omega}(a_X) - \alpha^{\Lambda_m}_{t;\omega}(a_X))T^B(t)\|\\
        &= \|(T^{\Lambda_n}(t))^* \alpha^{\Lambda_n, \Psi}_{t;\omega}(a_X)T^{\Lambda_n} (t) - T^{\Lambda_m}(t) \alpha^{\Lambda_m, \Psi}_{t;\omega}(a_X) T^{\Lambda_m}(t)\|\\
        &= \|\tilde \alpha^{\Lambda_n}_t (a_X) - \tilde \alpha^{\Lambda_m}(a_X)\|\,.
    \end{align*} So it is sufficient to show that the thermodynamic limit of $(\tilde \alpha^{\Lambda_n}_t)_{n=1}^\infty$ exists. It is easy to check that the $\tilde \alpha_t$ are generated by the local interaction 
        \[
            \tilde \Phi_\omega(Z,t) = (T^{Z}(t))^*\Phi(Z)( T^{Z}(t))\,,
        \] which we note obeys $\|\tilde \Phi_\omega(Z,t)\| = \|\Phi(Z)\|$ for any $t\in \R$, any $Z\Subset \Znu$, and all $\omega\in \Omega$. Hence $\sup_{t\in \R} \|\tilde \Phi(t)\|_F = \|\Phi\|_F<\infty$. Therefore, we can conclude the thermodynamic limit exists because of the time-inhomogeneous Lieb-Robinson bounds as in Theorem 3.5 \cite{NachtergaeleSimsYoung}. 
\end{proof}

   \begin{rmk}
    As a corollary, we obtain a different proof that models such as the disordered XY model studied in \cite{HamzaSimsStolz} or the XXZ chain \cite{Stolz_MBL} have deterministic bulk spectra (the usual proof is obtained by applying the Jordan-Wigner transformation and then the Pastur-Kirsch-Martinelli theorem). Simplicity of the eigenvalues for Anderson and Anderson-like models have been investigated in (e.g.) \cite{KleinMolchanov, NabokoNicholsStolz}. Interestingly, since both the XY and XXZ models have unique ground states, we have also shown that they are both ergodic. 
\end{rmk}
 
\appendix
\section{Measurability Concerns}
In this section, we will recall some basic facts about random variables and measures which take values in infinite-dimensional Banach spaces. For the convenience of the reader we present more details in this section as they will play a role in our applications of the material below.  

We will take a few lines to recall some basics about random variables in Banach space. Let $\mathcal{X}$ be a separable Banach space and let $(\Omega, \mathcal{F}, \P)$ be a probability space. Famously, there are two distinct notions of measurable functions from $\Omega$ into $\mathcal{X}$. One says that a function $s:\Omega \to \A$ is \textit{simple} if there exist measurable sets $E_1, E_2, \dots, E_k\in \mathcal{F}$ and elements $x_1, \dots, x_k\in \A$ so that 
    \[
        s_\omega = \sum_{j=1}^k x_j \cdot \chi_{E_j}(\omega)\,.
    \] A function $X:\Omega \to \mathcal{X}$ is \textit{strongly ($\P$-)measurable} if there exists a sequence of simple functions $s^{(n)}$ so that $\|X-s^{(n)}\|_{\mathcal{X}} \to 0$ $\P$-a.e. $\omega$. On the other hand, a function $Y:\Omega \to \mathcal{X}$ is said to be \textit{weakly measurable} if for every linear functional $f\in \mathcal{X}^*$, the $\C$-valued function $\Omega \ni \omega \longmapsto f(Y(\omega))$ is Borel measurable. If $\mathcal{X}$ is separable, these notions agree due to a general theorem of Pettis \cite{Pettis}: 

    \begin{thm}[Pettis Measurability Theorem]
        Let $\mathcal{X}$ be a separable Banach space, and let $(\Omega, \mathcal{F}, \P)$ be a probability space. Then a function $X:\Omega \to \mathcal{X}$ is strongly measurable if and only if it is weakly measurable.  
    \end{thm} For a proof see any of \cite{Bharucha-Reid, DiestelUhl, Hytonen_et_al_vol1}. With the Pettis Theorem in mind, we will refer to a measurable function into $\mathcal{X}$ as a(\textbf{n $\mathcal{X}$-valued}) \textbf{random variable}. 
 \begin{rmk}\label{rmk:norm_function}
        If $\mathcal{X}$ is a separable Banach space and $X:\Omega \to \mathcal{X}$ is a random variable, then the \textbf{norm function} given by $\Omega \ni \omega \mapsto \|X_\omega\|\in [0, \infty]$ is Borel measurable. Indeed since $\mathcal{X}$ is separable, there is a countable subset $D\subset \mathcal{X}^*$ so that $\|y\| = \sup_{f\in D}|f(y)|$ for any $y\in \mathcal{X}$ by standard functional analysis. 
    \end{rmk}
    It is standard theory that linear combinations of random variables in a Banach space are once again random variables. Furthermore, one can define $L^p$ spaces based on integrability of the norm function. These so-called \emph{Bochner spaces}  are defined for $1\le p <\infty$ as follows: 
    \[
    	L^p(\mathcal{X}) := \{ G:\Omega \to \mathcal{X} \text{ $\P$-measurable with } \|\, \|G\|_\mathcal{X} \|_{L^p(\P)} <\infty\}\,,
    \] where functions are identified modulo equality on set of positive measure. Each Bochner $L^p$ space is a Banach space in and of itself \cite[Theorem 3.7.8]{HillePhillips}. Recall that a given probability space $(\Omega, \mathcal{F}, \P)$ is \textbf{countably generated} if there is a sequence of measurable sets $F_n\in \mathcal{F}$ which generates $\mathcal{F}$ as a $\sigma$-algebra (cf. \cite[Definition 1.2.27]{Hytonen_et_al_vol1}). If $\mathcal{F}$ is countably generated, and $\mathcal{X}$ is separable then the Bochner spaces $L^p(\A)$ are separable Banach spaces by Proposition 1.2.29 of \cite{Hytonen_et_al_vol1}.   
    
    \subsection{Weak-$\ast$ Measurability and Duality}
    Now, let $\A$ be a separable $C^*$-algebra with a unit. The \emph{state space} of $\A$ is the set $\mathcal{S}(\A)\subset \A^*$ of linear functionals $\phi:\A \to \C$ which are positive $\phi(a^*a) \ge 0$ for all $a\in \A$, and unit-normalized $\phi(\one) = 1$. In this letter, we will be concerned with state-valued random variables $\Phi: \Omega \to \mathcal{S}(\A)\subset \mathcal{A}^*$ into the state space of $\A$. A natural topology on $\mathcal{S}(\A)$ with respect to which we should define measurability is the weak-$\ast$-topology (that is, the $\sigma(\mathcal{A}, \mathcal{A}^*)$ topology in the notation of \cite{ReedSimon}). There does not seem to be much in the literature about weakly-$\ast$ measurable functions in non-reflexive Banach spaces, so we will prove several basic properties here. Many of the following lemmas can be generalized to a separable Banach space with the appropriate properties, but we state the $C^*$-algebraic formulations for convenience. 

    \begin{define}
        Let $(\Omega, \mathcal{F}, \P)$ be a probability space and let $\A$ be a separable $C^*$-algebra. A mapping $\Phi:\Omega \to \A^*$ is \textbf{random linear functional} if for every $a\in \A$, one has that the function 
        \begin{equation}\label{eqn:w*_meas}
            \Omega \ni \omega \mapsto \Phi_\omega(a) \in \C\,,
        \end{equation} is Borel measurable in $\C$, and furthermore $\Phi$ is linear almost surely. 
    \end{define}
\begin{rmk}\label{rmk:composition_of_random_variables}
    There is a bit of a clash in terminology in the literature about analysis in Banach spaces which we need to point out. In the language of \cite{Bharucha-Reid, Hytonen_et_al_vol1, Hytonen_et_al_vol2, Hytonen_et_al_vol3} a mapping $\Phi$ as in~(\ref{eqn:w*_meas}) would be called a \emph{strongly $\P$-measurable} random linear operator when viewed as taking values in $\mathscr{L}(\A, \C)$. In particular, such a $\Phi$ takes random variables in $\A$ to random variables in $\C$, see \cite[Theorem 2.14]{Bharucha-Reid}. 
\end{rmk}

\begin{rmk}
    Note that if $\Phi_\omega$ is weak-$\ast$ measurable, then the operator-norm function $\omega\mapsto \|\Phi_\omega\|_{\A^*}$ is also measurable. Indeed, pick some countable dense subset of $\A$, in which case $\|\Phi_\omega\|_{\A^*} = \sup_n \frac{|\Phi_\omega(a_n)|}{\|a_n\|}$ with equality holding point-wise if we allow for $+\infty$. Cf. Remark~\ref{rmk:norm_function} above. 
\end{rmk}

The next lemma is not strictly necessary, but it may be of independent interest since it characterizes weak-$\ast$ measurability. In words, we show that weakly-$\ast$ measurable functions into the dual of a quasilocal algebra can be approximated by simple functions in a suitable topology. 

\begin{prop}\label{prop:w*_simple_functions}
    Suppose the one-site algebra $\A_{\{x\}}\cong \M_k$ for a fixed integer $k$ and let $\A_{\Znu}$ be the associated quasilocal algebra. Let $(\Omega, \mathcal{F}, \P)$ be a probability space. A function $f:\Omega \to \A_{\Znu}^*$ is a random linear functional if and only if there is a set $\Omega_0\subset \Omega$ with $\P\Omega_0 = 1$, and a sequence of simple functions $\ell^{(n)}_\omega$ so that for all $a\in \A_{\Znu}$, and all $\omega \in \Omega_0$, the limit $\lim_{n\to \infty} \ell^{(n)}_\omega(a) = f_\omega(a)$. 
\end{prop}

Before beginning the proof, we take a few lines to discuss the main idea of the proof. The key ingredient of Proposition~\ref{prop:w*_simple_functions} is that for any IAS $(\Lambda_n)_{n=1}^\infty$, there exist unital completely positive maps $\mathcal{E}_n:\A_{\Znu} \to \A_{\Lambda_n}$ (namely, taking partial traces), such that $\|a - \mathcal{E}_n(a)\| \to 0$ as $n\to \infty$ for any $a\in \A$ \cite{NachtergaeleScholzWerner,NachtergaeleSimsYoung}. By taking a weakly-$\ast$-measurable $f$, and restricting to $\A_{\Lambda_n}$, one can regard $f$ as a limit of weakly measurable functionals in the dual of the finite dimensional subalgebra $\A_{\Lambda_n}$. Thus, at each level $f|_{\A_{\Lambda_n}}$ is a limit of simple functions in $\A_{\Lambda_n}^*$. Then, by composing these restrictions with the $\mathcal{E}_n$ and taking a diagonal argument, one finds a sequence of simple functions in $\A_{\Znu}^*$ which weakly-$\ast$-approximates $f$.

In the $C^*$-algebras literature, this approximation property is an instance of the \emph{nuclearity} property. In words, a $C^*$-algebra $\A$ is \emph{nuclear} if the identity mapping $\id_{\A}$ approximately factors through finite dimensional matrix algebras strongly. More precisely, (see \cite[Proposition 2.2.6]{BrownOzawa})  $\A$ is nuclear if for every $n\in \N$, there is $k(n)\in \N$, and a pair of unital completely positive mappings $\pi_n:A\to \M_{k(n)}$ and $\iota_n:\M_{k(n)}\to A$ so that for any $a\in \A$, 
\begin{equation}\label{eqn:nuclear}
    \lim_{n\to \infty} \|\iota_n \circ \pi_n(a) - a\| = 0\,.
\end{equation}  Since this approximation property holds for a larger class of algebras and is relatively straightforward to state, we will prove Proposition~\ref{prop:w*_simple_functions} in this level of generality.

\begin{proof}[Proof of Proposition~\ref{prop:w*_simple_functions}]
    The backward implication has nothing to do with nuclearity. Instead we note that the hypothesis implies $\lim_{n\to \infty} |\phi^{(n)}_\omega(a)| = |f_\omega(a)|$, whence the seminorms defining the weak-$\ast$ topology pass through $f$ to become measurable functions in $\mathbb{R}$ since $|\phi^{(n)}_\omega(a)|$ is $\P$-measurable.   

    For the forward implication, let us first assume that the norm-function $\Omega \ni \omega \mapsto \|f_\omega\|\in [0, \infty)$ is essentially bounded (the general case will follow by truncation arguments). Write $F:= \esssup \|f_\omega\|<\infty$.  
    
    Now, find the sequence of matrix algebras $\M_{k(n)}$, and ucp maps $\pi_n:\A \to \M_{k(n)}$ and $\iota_n:\M_{k(n)} \to \A$ which satisfy the approximation property~(\ref{eqn:nuclear}). Let $f^{(n)}_\omega := f_\omega \circ \iota_n$. This $f^{(n)}_\omega$ can be viewed as a linear functional on the finite-dimensional subspace $ R_n:= \iota_n(\M_{k(n)}) \cong \M_{k(n)}$, where the isomorphism is that of Banach spaces. Furthermore, $f^{(n)}_\omega(\pi_n(a))$ is measurable for all $a\in \A$ by hypothesis. Therefore, we may view $f^{(n)}_\omega$ as a weakly-$\ast$ measurable random variable in the dual space $R_n^*$ viewed as a subspace of $\A^*$. 
    
    Since $R_n$ is finite dimensional, this means that there is a sequence of simple functions $\{s^{(n,m)}_\omega\}_{m=1}^\infty \subset (R_n)^*$ which approximate $f^{(n)}_\omega$ in the weak-$\ast$ topology of $R_n^*$ almost surely. Again using the fact that $R_n$ is finite dimensional, this is equivalent to convergence in the norm topology by basic functional analysis. Thus, there is a set $\Omega_n\subset \Omega$ with $\P[\Omega _n] =1$ so that
    \begin{equation}
        \lim_{m\to \infty} \|f^{(n)}_\omega - s^{(n,m)}_\omega\|_{R_n^*} = 0\quad \forall\omega \in \Omega_n\,.
    \end{equation} Notice that any simple function in $R_n^*$ can be canonically extended to a simple function in $\A^*$ since any finite-dimensional subspace of a Banach space is complemented \cite[Theorem 4.21(a)]{Rudin} (in particular, one need only compose with $\pi_n$). Therefore, for any $n$ there is $m_n$ so that 
    \[
        \|f^{(n)}_\omega - s^{(n, m_n)}_\omega \|_{\A^*} < 2^{-n}\quad \forall \omega \in \Omega_n\,.
    \] Put $\Omega_0 = \{\omega: \esssup \|f_\omega \| <\infty\} \cap \bigcap_{n=1}^\infty \Omega_n,$ and let $\ell^{(n)}_\omega := \chi_{\Omega_0}(\omega) s^{(n, m_n)}_\omega \circ \pi_n$. Notice $\P(\Omega_0) = 1$, and each $\ell^{(n)}_\omega$ is a simple function in $\A^*$. 

    Now, let $\epsilon>0$,  $a\in \A$, and $\omega \in \Omega_0$. By nuclearity, there exists $n_0(a,\epsilon, \omega)$ so that for all $n\ge n_0$, we have $\|a - \iota_n \circ \pi_n (a)\| < \epsilon$. Furthermore, the following inequalities hold point-wise in $\Omega_0$:
    \begin{align*}
        0\le |f_\omega(a) - \ell^{(n)}_\omega(a)| 
                &\le \|f_\omega\| \epsilon + |f_\omega \circ \iota_n \circ \pi_n (a) - \chi_{\Omega_0}(\omega) s^{(n,m_n)}_\omega(\pi_n(a))|\\
        &\le \epsilon\, F + \|f^{(n)}_\omega - s^{(n,m_n)}_\omega \|_{R_n^*}\, \|\pi_n(a)\|\le \epsilon\, F + 2^{-n}\|a\|\,,
    \end{align*} where the last inequality follows from the definition of $\ell^{(n)}$ and the fact that ucp maps are completely contractive. 

    Now, suppose that $\Omega \ni \omega \mapsto \|f_\omega\|\in [0, \infty)$ is not essentially bounded. Let $B_k = \{ \omega \colon \|f_\omega\|\le k\}$ which are measurable. By applying the above argument to the random variable $\chi_{B_k}(\omega) \cdot f_\omega$, we find a sequence of simple functions $\ell^{(n,k)}_\omega$ which almost surely approximate the truncation of $f_\omega$. Applying a diagonal argument once more, we can construct a sequence which approximates $f_\omega$ almost surely in the weak-$\ast$ topology. 
\end{proof}

We restrict our focus to those weakly-$\ast$ measurable functions which have essentially bounded norm-functions.
\begin{note}
	Let $\A$ be a separable $C^*$-algebra with a unit and let $(\Omega, \mathcal{F}, \P)$ be a probability space. Let $\mathcal{L}^\infty _{w^*}(\A^*)$ denote the vector space of weakly-$\ast$-measurable functions $\Phi:\Omega \to \A^*$ with essentially bounded norm function $(\Omega \ni \omega \mapsto \|\Phi_\omega\|_{\A^*}) \in L^{\infty}(\P)$. 
\end{note}

It is a well-known fact that any operator in a $C^*$-algebra admits a Cartesian decomposition into a sum of positive semidefinite operators. In fact this induces an ordering on $\A$, which in turn induces an ordering on $\A^*$. Recall that a linear functional $\varphi$ in $\A^*$ is \emph{positive} if $\varphi(a^*a) \ge 0$ for all $a\in \A$. Furthermore, the Russo-Dye theorem states that any positive linear functional attains its operator norm at the unit, namely $\|\varphi\|_{\A^*} = |\varphi(\one)|$ \cite[Proposition 3.6]{Paulsen}. 

A result originally due to Han\v{s} \cite[Theorem 2.14]{Bharucha-Reid} ensures that if $\Phi_\omega$ is a weakly-$\ast$-measurable random variable and $A\in L^1(\A)$ then the composition $\omega \mapsto \Phi_\omega \circ A_\omega \in \C$ is Borel measurable. Therefore, if $\Phi_\omega$ is weakly-$\ast$ measurable and takes values in $(\A^*)_+$, it is essentially bounded in norm if and only if $0\le \Phi(\one)\in L^\infty(\P)$. We will say that a function $A\in L^1(\A)$ is positive if $A(\omega) \in \A_+$ almost surely. Then an element $\Phi$ of $[L^1(\A)]^*$ is \emph{positive} if whenever $A\in (L^1(\A))_+$ one has $\Phi\circ A \ge 0$ with probability one. 

 In fact by definition, on a set with full probability $|\Phi_\omega (A_\omega)| \le \|\Phi_\omega\|_{\A^*} \|A_\omega\|_{\A}$ for arbitrary $\Phi$ and $A$. A quick application of H\"older's inequality yields the following lemma.

\begin{lem}\label{lem:pairing}
	Let $(\Omega, \mathcal{F}, \P)$ be a probability space and $\A$ be a separable $C^*$-algebra with a unit. If $\Phi\in \mathcal{L}^\infty_{w^*}(\A)$ and $A\in L^1(\A)$, then $\Phi \circ A \in L^1(\P)$ and the following estimate holds: 
	\begin{equation}
	\left| \int_\Omega \Phi_\omega (A_\omega) d\P\right| \le \|\,\|\Phi\|_{\A^*}\|_{L^\infty(\P)}\, \|A\|_{L^1(\A)}\,.
	\end{equation} Therefore, the pairing $(\Phi,A):= \int_\Omega \Phi\circ A d\P$ defines a bounded bilinear pairing which is faithful on positive entries. 
\end{lem}

Recall the following elementary lemma from measure theory. 
\begin{lem}\label{lem:essbddRND}
	Suppose $(\Omega, \mathcal{F})$ is a measurable space, $\nu$ is a complex measure in $\Omega$ and $\mu$ is a finite positive measure on $\Omega$. The following are equivalent: 
	\begin{enumerate}
		\item  The total variation $|\nu| \ll \mu$ and the Radon-Nikodym derivative $\frac{d\mu}{d\lambda} \in L^\infty(\mu)$, 
		\item  there is a constant $C>0$ so that $|\nu|(E)\le C \mu(E)$ for all $E\in \mathcal{F}$, and in particular $\|\frac{d\nu}{d\mu}\|_{L^\infty(\mu)}\le C$. 
	\end{enumerate}
\end{lem} 

\begin{prop}\label{prop:w*RND}
	Suppose $(\Omega, \mathcal{F}, \P)$ is a probability space and $\A$ is a separable $C^*$-algebra with a unit. Consider the Bochner space $L^1(\A)$. Suppose $\Phi \in [L^1(\A)]^*$ is positive. Then there is a weak-$\ast$ measurable function $G\in \mathcal{L}^\infty_{w^*}(\A^*)$ so that $\P[G_\omega(a^*a) \ge 0] =1$ for all $a\in \A$ and for any $A\in L^1(\A)$, one has
	\begin{equation}\label{eqn:w*RND2}
		\Phi(A) = \int_\Omega G_\omega(A(\omega)) d\P\,.
	\end{equation} In particular, $\|\, \|G\|_{\A^*}\, \|_{L^\infty(\P)} \le \|\Phi\|_{L^1(\A)^*}$. 
\end{prop} 

Before we prove this result, a few remarks are in order. \begin{rmk}\label{rmk:RND_1}
    We emphasize that we have only identified $L^1(\A)^*$ with $\mathcal{L}^\infty_{w^*}(\A^*)$ in a weak-topological sense. To wit, there are examples of Banach spaces $\mathcal{X}$ where the dual $L^1(\mathcal{X})^*$ is strictly larger than essentially bounded functions, and hence is \emph{not} isometric to $L^\infty(\mathcal{X})$. This is related to the question of whether or not a vector measure possesses a Radon-Nikodym derivative is a property of both the measure itself \emph{and} the target Banach space. Banach spaces with the property that every vector measure admits an appropriate Radon-Nikodym derivative are said to possess the Radon Nikodym Property (RNP), and there are known Banach spaces for which this fails \cite{Ryan}.
\end{rmk}
\begin{rmk}\label{rmk:RND_2}
    One might hope that the dual of a quasilocal algebra might have the RNP, but this is false. Work in the 1970's by Jensen \cite{Jensen_I, Jensen_II}, and later by Chu \cite[Theorem 3]{Chu} shows that the dual $\B^*$ of a $C^*$ algebra, $\B$, has the RNP if and only if $\B$ is a so-called \emph{scattered $C^*$-algebra}. Work by Kusada \cite[Theorem 2.2]{Kusuda} shows that a necessary condition for $\B$ to be scattered is that it must be a type I $C^*$-algebra: meaning every nonzero irreducible representation of $\B$ contains the compact operators \cite[p. 169]{Murphy}. By contrast, most examples of quasi-local algebras in the mathematical physics literature are Uniformly Hyperfinite (UHF) algebras which are known to admit irreducible type $\mathrm{II}_1$ representations (namely, the GNS-representation with respect to the unique tracial state, cf. \cite[Theorem 11.2.2]{Anatharaman-DelarochePopa}). Thus, for many cases of interest, the dual of the quasilocal algebra fails to have the RNP, and one must check the individual probability measure for this property instead. 
 \end{rmk}
\begin{proof}[Proof of Proposition~\ref{prop:w*RND}]
	Fix a positive element $a\in \A_+$ and define the set function 
	\begin{equation}\label{eqn:induced_measure}
		\mathcal{F}\ni E\longmapsto \nu_{\Phi}(E;a) := \Phi(\chi_E\, a)\in [0, \infty)\,.
	\end{equation} Then, $\nu_{\Phi}$ is clearly finitely additive, and linear in the operator-entry. To see that $\nu_{\Phi}(\,\cdot\,;a)$ is countably additive, recall that the countable additivity of $\P$ can be rephrased as the $L^1(\P)$ convergence of characteristic functions: whenever $E= \bigsqcup E_n$, the limit $\|\chi_{E} - \sum_1^N \chi_{E_n}\|_{L^1(\P)} \to 0$ as $N\to \infty$. Therefore, the sequence of simple functions $\sum_1^N \chi_{E_n} a \to \chi_E a$ in $L^1(\A)$. In fact, 
	\[
		\left| \sum_{n=1}^N \nu_{\Phi}(E_n;a) - \nu_{\Phi}(E; a) \right| \le \|a \|\|\Phi\| \left\|\chi_{E} - \sum_1^N \chi_{E_n}\right\|_{L^1(\P)}\,.
	\] Now, notice that by construction
	\[
		0\le \nu_{\Phi}(E;a)\le \|a\| \|\Phi\|\, \P[E],
	\] whence $\nu_{\Phi}(\,\cdot\,;a)\ll \P$ and by Lemma~\ref{lem:essbddRND}, there is a positive scalar function $g_a = \frac{d\nu_{\Phi}(\cdot; a)}{d\P}\in L^\infty(\P)$ so that 
	\[
		\nu_{\Phi}(E; a) = \int_E g_a(\omega) d\P\,.
	\] Since the scalar Radon-Nikodym derivative is $\P$-essentially unique, the mapping $a\mapsto g_a(\omega)$ may be extended to a linear mapping almost everywhere. Note we are crucially using separability of $\A$ in order to identify the full probability event in $\Omega$ where linearity holds.  By Lemma~\ref{lem:essbddRND}, there is a set $\Omega_0\subset \Omega$ with full probability so that $|g_a(\omega)| \le \|a\| \|\Phi\|$ for all $\omega \in \Omega_0$ since $\|g_a\|_{L^\infty(\P)}\le \|\Phi\| \|a\|$. By the Uniform Boundedness Principle, we conclude $\sup_{\omega\in \Omega_0} \|g_{\bullet}(\omega)\|_{\A^*} \le \|\Phi\|$ whence $g\in \mathcal{L}^\infty_{w^*}(\A^*)$ as required. Equation~(\ref{eqn:w*RND2}) is a straightforward application of approximation by simple functions. 
\end{proof}

\bibliographystyle{plain}
\bibliography{refs}

@book {AizenmannWarzel,
    AUTHOR = {Aizenman, M. and Warzel, S.},
     TITLE = {Random operators},
    SERIES = {Graduate Studies in Mathematics},
    VOLUME = {168},
      NOTE = {Disorder effects on quantum spectra and dynamics},
 PUBLISHER = {American Mathematical Society, Providence, RI},
      YEAR = {2015},
     PAGES = {xiv+326},
      ISBN = {978-1-4704-1913-4},
   MRCLASS = {82B44 (46N50 47B80 60H25 81Q10 81Q12 82B10 82D30)},
  MRNUMBER = {3364516},
       DOI = {10.1090/gsm/168},
       URL = {https://doi.org/10.1090/gsm/168},
}

@unpublished{Anatharaman-DelarochePopa,
    author = {Anatharaman-Delaroche, C. and Popa, S.},
    title = {An introduction to {$\mathrm{II}_1$} factors},
    note = {Unpublished manuscript hosted by UCLA.},
    url = {https://www.math.ucla.edu/~popa/Books/IIun.pdf}
}

@book {Bharucha-Reid,
    AUTHOR = {Bharucha-Reid, A. T.},
     TITLE = {Random integral equations},
    SERIES = {Mathematics in Science and Engineering},
    VOLUME = {Vol. 96},
 PUBLISHER = {Academic Press, New York-London},
      YEAR = {1972},
     PAGES = {xiii+267},
   MRCLASS = {60H20 (45F05)},
  MRNUMBER = {443086},
MRREVIEWER = {J.\ S.\ Milton},
}

@book {BratteliRobinsonI,
    AUTHOR = {Bratteli, O. and Robinson, D. W.},
     TITLE = {Operator algebras and quantum statistical mechanics. 1},
    SERIES = {Texts and Monographs in Physics},
   EDITION = {Second},
      NOTE = {$C^\ast$- and $W^\ast$-algebras, symmetry groups,
              decomposition of states},
 PUBLISHER = {Springer-Verlag, New York},
      YEAR = {1987},
     PAGES = {xiv+505},
      ISBN = {0-387-17093-6},
   MRCLASS = {46Lxx (81-02 82-02)},
  MRNUMBER = {887100},
       DOI = {10.1007/978-3-662-02520-8},
       URL = {https://doi.org/10.1007/978-3-662-02520-8},
}

@book {BratteliRobinsonII,
    AUTHOR = {Bratteli, O. and Robinson, D. W.},
     TITLE = {Operator algebras and quantum statistical mechanics. 2},
    SERIES = {Texts and Monographs in Physics},
   EDITION = {Second},
      NOTE = {Equilibrium states. Models in quantum statistical mechanics},
 PUBLISHER = {Springer-Verlag, Berlin},
      YEAR = {1997},
     PAGES = {xiv+519},
      ISBN = {3-540-61443-5},
   MRCLASS = {82B10 (46Lxx 81S05 82-02 82C10)},
  MRNUMBER = {1441540},
       DOI = {10.1007/978-3-662-03444-6},
       URL = {https://doi.org/10.1007/978-3-662-03444-6},
}

@book {BrownOzawa,
    AUTHOR = {Brown, N. P. and Ozawa, N.},
     TITLE = {{$C^*$}-algebras and finite-dimensional approximations},
    SERIES = {Graduate Studies in Mathematics},
    VOLUME = {88},
 PUBLISHER = {American Mathematical Society, Providence, RI},
      YEAR = {2008},
     PAGES = {xvi+509},
      ISBN = {978-0-8218-4381-9; 0-8218-4381-8},
   MRCLASS = {46L05 (43A07 46-02 46L10)},
  MRNUMBER = {2391387},
MRREVIEWER = {Mikael\ R\o rdam},
       DOI = {10.1090/gsm/088},
       URL = {https://doi.org/10.1090/gsm/088},
}

@book {Conway_FA,
    AUTHOR = {Conway, J.},
     TITLE = {A course in functional analysis},
    SERIES = {Graduate Texts in Mathematics},
    VOLUME = {96},
   EDITION = {Second},
 PUBLISHER = {Springer-Verlag, New York},
      YEAR = {1990},
     PAGES = {xvi+399},
      ISBN = {0-387-97245-5},
   MRCLASS = {46-01 (47-01)},
  MRNUMBER = {1070713},
}

@book {DiestelUhl,
    AUTHOR = {Diestel, J. and Uhl, Jr., J. J.},
     TITLE = {Vector measures},
    SERIES = {Mathematical Surveys},
    VOLUME = {No. 15},
      NOTE = {With a foreword by B. J. Pettis},
 PUBLISHER = {American Mathematical Society, Providence, RI},
      YEAR = {1977},
     PAGES = {xiii+322},
   MRCLASS = {28A45 (46B05 46G10)},
  MRNUMBER = {453964},
MRREVIEWER = {Robert\ E.\ Huff},
}

@book {HillePhillips,
    AUTHOR = {Hille, E. and Phillips, R. S.},
     TITLE = {Functional analysis and semi-groups},
    SERIES = {American Mathematical Society Colloquium Publications},
    VOLUME = {Vol. 31},
      NOTE = {rev. ed},
 PUBLISHER = {American Mathematical Society, Providence, RI},
      YEAR = {1957},
     PAGES = {xii+808},
   MRCLASS = {46.2X},
  MRNUMBER = {89373},
MRREVIEWER = {M.\ H.\ Stone},
}

@book {Hytonen_et_al_vol1,
    AUTHOR = {Hyt\"onen, T. and van Neerven, J. and Veraar, M. and Weis, L.},
     TITLE = {Analysis in {B}anach spaces. {V}ol. {I}. {M}artingales and
              {L}ittlewood-{P}aley theory},
    SERIES = {Ergebnisse der Mathematik und ihrer Grenzgebiete. 3. Folge. A
              Series of Modern Surveys in Mathematics [Results in
              Mathematics and Related Areas. 3rd Series. A Series of Modern
              Surveys in Mathematics]},
    VOLUME = {63},
 PUBLISHER = {Springer, Cham},
      YEAR = {2016},
     PAGES = {xvi+614},
      ISBN = {978-3-319-48519-5; 978-3-319-48520-1},
   MRCLASS = {46-02 (42B35 46E30)},
  MRNUMBER = {3617205},
MRREVIEWER = {Adam\ Os\polhk ekowski},
}

@book {Hytonen_et_al_vol2,
    AUTHOR = {Hyt\"onen, T. and van Neerven, J. and Veraar, M. and Weis, L.},
     TITLE = {Analysis in {B}anach spaces. {V}ol. {II}. {P}robabilistic
              methods and operator theory},
    SERIES = {Ergebnisse der Mathematik und ihrer Grenzgebiete. 3. Folge. A
              Series of Modern Surveys in Mathematics [Results in
              Mathematics and Related Areas. 3rd Series. A Series of Modern
              Surveys in Mathematics]},
    VOLUME = {67},
 PUBLISHER = {Springer, Cham},
      YEAR = {2017},
     PAGES = {xxi+616},
      ISBN = {978-3-319-69807-6; 978-3-319-69808-3},
   MRCLASS = {46-02 (42B35 46E30 47-02 60B11 60H30)},
  MRNUMBER = {3752640},
MRREVIEWER = {Adam\ Os\polhk ekowski},
       DOI = {10.1007/978-3-319-69808-3},
       URL = {https://doi.org/10.1007/978-3-319-69808-3},
}

@book {Hytonen_et_al_vol3,
    AUTHOR = {Hyt\"onen, T. and van Neerven, J. and Veraar, M. and  Weis, L.},
     TITLE = {Analysis in {B}anach spaces. {V}ol. {III}. {H}armonic analysis
              and spectral theory},
    SERIES = {Ergebnisse der Mathematik und ihrer Grenzgebiete. 3. Folge. A
              Series of Modern Surveys in Mathematics [Results in
              Mathematics and Related Areas. 3rd Series. A Series of Modern
              Surveys in Mathematics]},
    VOLUME = {76},
 PUBLISHER = {Springer, Cham},
      YEAR = {[2023] \copyright 2023},
     PAGES = {xxi+826},
      ISBN = {978-3-031-46597-0; 978-3-031-46598-7},
   MRCLASS = {46-02 (35Kxx 42Bxx 46Bxx 46Exx)},
  MRNUMBER = {4696978},
MRREVIEWER = {Pierre\ Portal},
       DOI = {10.1007/978-3-031-46598-7},
       URL = {https://doi.org/10.1007/978-3-031-46598-7},
}

@book {Murphy,
    AUTHOR = {Murphy, G. J.},
     TITLE = {{$C^*$}-algebras and operator theory},
 PUBLISHER = {Academic Press, Inc., Boston, MA},
      YEAR = {1990},
     PAGES = {x+286},
      ISBN = {0-12-511360-9},
   MRCLASS = {46Lxx (46-01)},
  MRNUMBER = {1074574},
MRREVIEWER = {E.\ Gerlach},
}

@book {Paulsen,
    AUTHOR = {Paulsen, V.},
     TITLE = {Completely bounded maps and operator algebras},
    SERIES = {Cambridge Studies in Advanced Mathematics},
    VOLUME = {78},
 PUBLISHER = {Cambridge University Press, Cambridge},
      YEAR = {2002},
     PAGES = {xii+300},
      ISBN = {0-521-81669-6},
   MRCLASS = {46L07 (47A20 47L30)},
  MRNUMBER = {1976867},
MRREVIEWER = {Christian\ Le Merdy},
}

@book {Ryan,
    AUTHOR = {Ryan, R. A.},
     TITLE = {Introduction to tensor products of {B}anach spaces},
    SERIES = {Springer Monographs in Mathematics},
 PUBLISHER = {Springer-Verlag London, Ltd., London},
      YEAR = {2002},
     PAGES = {xiv+225},
      ISBN = {1-85233-437-1},
   MRCLASS = {46B28 (46-01 47A07 47B10 47L20)},
  MRNUMBER = {1888309},
MRREVIEWER = {Andreas\ Defant},
       DOI = {10.1007/978-1-4471-3903-4},
       URL = {https://doi.org/10.1007/978-1-4471-3903-4},
}

@book {ReedSimon,
    AUTHOR = {Reed, M. and Simon, B.},
     TITLE = {Methods of modern mathematical physics. {I}. {F}unctional
              analysis},
 PUBLISHER = {Academic Press, New York-London},
      YEAR = {1972},
     PAGES = {xvii+325},
   MRCLASS = {47-02 (81.47)},
  MRNUMBER = {493419},
MRREVIEWER = {P.\ R.\ Chernoff},
}

@book {Rudin,
    AUTHOR = {Rudin, W.},
     TITLE = {Functional Analysis},
    SERIES = {International Series in Pure and Applied Mathematics},
   EDITION = {Second},
 PUBLISHER = {McGraw-Hill, Inc., New York},
      YEAR = {1991},
     PAGES = {xviii+424},
      ISBN = {0-07-054236-8},
   MRCLASS = {46-01 (47-01)},
  MRNUMBER = {1157815},
}

@book {Weidmann,
    AUTHOR = {Weidmann, J.},
     TITLE = {Linear operators in {H}ilbert spaces},
    SERIES = {Graduate Texts in Mathematics},
    VOLUME = {68},
      NOTE = {Translated from the German by Joseph Sz\"ucs},
 PUBLISHER = {Springer-Verlag, New York-Berlin},
      YEAR = {1980},
     PAGES = {xiii+402},
      ISBN = {0-387-90427-1},
   MRCLASS = {47-01 (46-01)},
  MRNUMBER = {566954},
}

@article {AKLT,
    AUTHOR = {Affleck, I. and Kennedy, T. and Lieb, E. H. and Tasaki,
              H.},
     TITLE = {Valence bond ground states in isotropic quantum
              antiferromagnets},
   JOURNAL = {Comm. Math. Phys.},
  FJOURNAL = {Communications in Mathematical Physics},
    VOLUME = {115},
      YEAR = {1988},
    NUMBER = {3},
     PAGES = {477--528},
      ISSN = {0010-3616,1432-0916},
   MRCLASS = {82A15 (82A68)},
  MRNUMBER = {931672},
MRREVIEWER = {Claus\ Montonen},
       URL = {http://projecteuclid.org/euclid.cmp/1104161001},
}

@article {Abdul-Rahman_et_al,
    AUTHOR = {Abdul-Rahman, H. and Nachtergaele, B. and Sims, R.
              and Stolz, G.},
     TITLE = {Localization properties of the disordered {XY} spin chain: a
              review of mathematical results with an eye toward many-body
              localization},
   JOURNAL = {Ann. Phys.},
  FJOURNAL = {Annalen der Physik},
    VOLUME = {529},
      YEAR = {2017},
    NUMBER = {7},
     PAGES = {201600280, 17},
      ISSN = {0003-3804,1521-3889},
   MRCLASS = {82B10 (81V35)},
  MRNUMBER = {3671049},
MRREVIEWER = {Benjamin\ Lees},
       DOI = {10.1002/andp.201600280},
       URL = {https://doi.org/10.1002/andp.201600280},
}

@article{Anderson,
  title = {Absence of Diffusion in Certain Random Lattices},
  author = {Anderson, P. W.},
  journal = {Phys. Rev.},
  volume = {109},
  issue = {5},
  pages = {1492--1505},
  numpages = {0},
  year = {1958},
  month = {Mar},
  publisher = {American Physical Society},
  doi = {10.1103/PhysRev.109.1492},
  url = {https://link.aps.org/doi/10.1103/PhysRev.109.1492}
}

@article{Baldwin,
  title = {Subballistic operator growth in spin chains with heavy-tailed random fields},
  author = {Baldwin, C. L.},
  journal = {Phys. Rev. B},
  volume = {111},
  issue = {18},
  pages = {184204},
  numpages = {14},
  year = {2025},
  month = {May},
  publisher = {American Physical Society},
  doi = {10.1103/PhysRevB.111.184204},
  url = {https://link.aps.org/doi/10.1103/PhysRevB.111.184204}
}

@misc{Bermudezet_al,
      title={Proofs for Folklore Theorems on the Radon-Nikodym Derivative}, 
      author={Bermudez, Y. and Bisson, G. and Esnaola, I. and Perlaza, S. M. },
      year={2025},
      eprint={2501.18374},
      archivePrefix={arXiv},
      primaryClass={cs.IT},
      url={https://arxiv.org/abs/2501.18374}, 
      note = {arXiv:2501.18374},
}

@article {BravyiHastingsMichalakis,
    AUTHOR = {Bravyi, S. and Hastings, M. B. and Michalakis,
              S.},
     TITLE = {Topological quantum order: stability under local
              perturbations},
   JOURNAL = {J. Math. Phys.},
  FJOURNAL = {Journal of Mathematical Physics},
    VOLUME = {51},
      YEAR = {2010},
    NUMBER = {9},
     PAGES = {093512, 33},
      ISSN = {0022-2488,1089-7658},
   MRCLASS = {82B10 (47A55 47N50 82B20 82B26)},
  MRNUMBER = {2742836},
MRREVIEWER = {Vladimir\ V.\ Kisil},
       DOI = {10.1063/1.3490195},
       URL = {https://doi.org/10.1063/1.3490195},
}

@article {Chu,
    AUTHOR = {Chu, C.-H.},
     TITLE = {A note on scattered {$C\sp{\ast} $}-algebras and the
              {R}adon-{N}ikod\'ym property},
   JOURNAL = {J. London Math. Soc. (2)},
  FJOURNAL = {Journal of the London Mathematical Society. Second Series},
    VOLUME = {24},
      YEAR = {1981},
    NUMBER = {3},
     PAGES = {533--536},
      ISSN = {0024-6107,1469-7750},
   MRCLASS = {46L05 (46L10)},
  MRNUMBER = {635884},
MRREVIEWER = {Shinz\B o\ Kawamura},
       DOI = {10.1112/jlms/s2-24.3.533},
       URL = {https://doi.org/10.1112/jlms/s2-24.3.533},
}

@misc{EkbladMoreno-NadalesRoonSchenker, 
    title = {Parent {H}amiltonians for Ergodic Matrix Product States},
    author = {Ekblad, O. and Moreno-Nadales, E. and Roon, E. B. and Schenker, J. H.},
    year = {2026},
    note = {\emph{Forthcoming}},
}

@article{ElgartKlein25,
    AUTHOR = {Elgart, A. and Klein, A.},
     TITLE = {Localization phenomena in the random {XXZ} spin chain},
   JOURNAL = {J. Funct. Anal.},
  FJOURNAL = {Journal of Functional Analysis},
    VOLUME = {290},
      YEAR = {2026},
    NUMBER = {7},
     PAGES = {Paper No. 111320, 41},
      ISSN = {0022-1236,1096-0783},
   MRCLASS = {82B44 (47B80 60H25 81Q10 82C44)},
  MRNUMBER = {5008480},
       DOI = {10.1016/j.jfa.2025.111320},
       URL = {https://doi.org/10.1016/j.jfa.2025.111320},
}

@article {ElgartKlein24_1,
    AUTHOR = {Elgart, A. and Klein, A.},
     TITLE = {Slow propagation of information on the random {XXZ} quantum
              spin chain},
   JOURNAL = {Comm. Math. Phys.},
  FJOURNAL = {Communications in Mathematical Physics},
    VOLUME = {405},
      YEAR = {2024},
    NUMBER = {10},
     PAGES = {Paper No. 239, 27},
      ISSN = {0010-3616,1432-0916},
   MRCLASS = {82C44 (47N50 60K35 82C10)},
  MRNUMBER = {4797750},
MRREVIEWER = {Jules\ Lamers},
       DOI = {10.1007/s00220-024-05127-y},
       URL = {https://doi.org/10.1007/s00220-024-05127-y},
}

@article {ElgartKlein24_2,
    AUTHOR = {Elgart, A. and Klein, A.},
     TITLE = {Localization in the random {XXZ} quantum spin chain},
   JOURNAL = {Forum Math. Sigma},
  FJOURNAL = {Forum of Mathematics. Sigma},
    VOLUME = {12},
      YEAR = {2024},
     PAGES = {Paper No. e129, 33},
      ISSN = {2050-5094},
   MRCLASS = {82B44 (47B80 60H25 81Q10 82C44)},
  MRNUMBER = {4845951},
       DOI = {10.1017/fms.2024.119},
       URL = {https://doi.org/10.1017/fms.2024.119},
}

@misc{ElgartKlein_2026,
      title={Many-body localization for the random XXZ spin chain in fixed energy intervals}, 
      author={Elgart, A.  and Klein, A.},
      year={2026},
      eprint={2602.01441},
      archivePrefix={arXiv},
      primaryClass={math-ph},
      url={https://arxiv.org/abs/2602.01441}, 
      note = {arXiv:2602.01441}
}

@article{FannesNachtergaeleWerner,
	Author = {Fannes, M. and Nachtergaele, B. and Werner, R. F.},
	Journal = {Communications in Mathematical Physics},
	Number = {3},
	Pages = {443--490},
	Title = {Finitely correlated states on quantum spin chains},
	Volume = {144},
	Year = {1992}
}

@article {HamzaSimsStolz,
    AUTHOR = {Hamza, E. and Sims, R. and Stolz, G.},
     TITLE = {Dynamical localization in disordered quantum spin systems},
   JOURNAL = {Comm. Math. Phys.},
  FJOURNAL = {Communications in Mathematical Physics},
    VOLUME = {315},
      YEAR = {2012},
    NUMBER = {1},
     PAGES = {215--239},
      ISSN = {0010-3616,1432-0916},
   MRCLASS = {82B44},
  MRNUMBER = {2966945},
       DOI = {10.1007/s00220-012-1544-6},
       URL = {https://doi.org/10.1007/s00220-012-1544-6},
}

@article {HastingsKoma,
    AUTHOR = {Hastings, M. B. and Koma, T.},
     TITLE = {Spectral gap and exponential decay of correlations},
   JOURNAL = {Comm. Math. Phys.},
  FJOURNAL = {Communications in Mathematical Physics},
    VOLUME = {265},
      YEAR = {2006},
    NUMBER = {3},
     PAGES = {781--804},
      ISSN = {0010-3616,1432-0916},
   MRCLASS = {82B10},
  MRNUMBER = {2231689},
MRREVIEWER = {Elena\ A.\ Zhizhina},
       DOI = {10.1007/s00220-006-0030-4},
       URL = {https://doi.org/10.1007/s00220-006-0030-4},
}

@article {Imbrie,
    AUTHOR = {Imbrie, J. Z.},
     TITLE = {On many-body localization for quantum spin chains},
   JOURNAL = {J. Stat. Phys.},
  FJOURNAL = {Journal of Statistical Physics},
    VOLUME = {163},
      YEAR = {2016},
    NUMBER = {5},
     PAGES = {998--1048},
      ISSN = {0022-4715,1572-9613},
   MRCLASS = {82B44},
  MRNUMBER = {3493184},
       DOI = {10.1007/s10955-016-1508-x},
       URL = {https://doi.org/10.1007/s10955-016-1508-x},
}

@article{JauslinLemm,
  doi = {10.22331/q-2022-09-01-790},
  url = {https://doi.org/10.22331/q-2022-09-01-790},
  title = {Random translation-invariant {H}amiltonians and their spectral gaps},
  author = {Jauslin, I. and Lemm, M.},
  journal = {{Quantum}},
  issn = {2521-327X},
  publisher = {{Verein zur F{\"{o}}rderung des Open Access Publizierens in den Quantenwissenschaften}},
  volume = {6},
  pages = {790},
  month = sep,
  year = {2022}
}

@article {Jensen_I,
    AUTHOR = {Jensen, H. E.},
     TITLE = {Scattered {$C\sp*$}-algebras},
   JOURNAL = {Math. Scand.},
  FJOURNAL = {Mathematica Scandinavica},
    VOLUME = {41},
      YEAR = {1977},
    NUMBER = {2},
     PAGES = {308--314},
      ISSN = {0025-5521,1903-1807},
   MRCLASS = {46L05},
  MRNUMBER = {482242},
MRREVIEWER = {Shinz\B o\ Kawamura},
       DOI = {10.7146/math.scand.a-11723},
       URL = {https://doi.org/10.7146/math.scand.a-11723},
}

@article {Jensen_II,
    AUTHOR = {Jensen, H. E.},
     TITLE = {Scattered {$C\sp{\ast} $}-algebras. {II}},
   JOURNAL = {Math. Scand.},
  FJOURNAL = {Mathematica Scandinavica},
    VOLUME = {43},
      YEAR = {1978},
    NUMBER = {2},
     PAGES = {308--310},
      ISSN = {0025-5521,1903-1807},
   MRCLASS = {46L05},
  MRNUMBER = {531308},
MRREVIEWER = {Shinz\B o\ Kawamura},
       DOI = {10.7146/math.scand.a-11782},
       URL = {https://doi.org/10.7146/math.scand.a-11782},
}

@article {KleinMolchanov,
    AUTHOR = {Klein, A. and Molchanov, S.},
     TITLE = {Simplicity of eigenvalues in the {A}nderson model},
   JOURNAL = {J. Stat. Phys.},
  FJOURNAL = {Journal of Statistical Physics},
    VOLUME = {122},
      YEAR = {2006},
    NUMBER = {1},
     PAGES = {95--99},
      ISSN = {0022-4715,1572-9613},
   MRCLASS = {82B44 (47A10 47B80 47N55)},
  MRNUMBER = {2203783},
MRREVIEWER = {Nariyuki\ Minami},
       DOI = {10.1007/s10955-005-8009-7},
       URL = {https://doi.org/10.1007/s10955-005-8009-7},
}

@article {KirschMartinelli,
    AUTHOR = {Kirsch, W. and Martinelli, F.},
     TITLE = {On the ergodic properties of the spectrum of general random
              operators},
   JOURNAL = {J. Reine Angew. Math.},
  FJOURNAL = {Journal f\"ur die Reine und Angewandte Mathematik. [Crelle's
              Journal]},
    VOLUME = {334},
      YEAR = {1982},
     PAGES = {141--156},
      ISSN = {0075-4102,1435-5345},
   MRCLASS = {60H25 (35P05 35R60 81C10 82A05)},
  MRNUMBER = {667454},
MRREVIEWER = {M.\ Fukushima},
       DOI = {10.1515/crll.1982.334.141},
       URL = {https://doi.org/10.1515/crll.1982.334.141},
}

@article {Knabe,
    AUTHOR = {Knabe, S.},
     TITLE = {Energy gaps and elementary excitations for certain
              {VBS}-quantum antiferromagnets},
   JOURNAL = {J. Statist. Phys.},
  FJOURNAL = {Journal of Statistical Physics},
    VOLUME = {52},
      YEAR = {1988},
    NUMBER = {3-4},
     PAGES = {627--638},
      ISSN = {0022-4715,1572-9613},
   MRCLASS = {82A68},
  MRNUMBER = {968951},
MRREVIEWER = {Muthiah\ Daniel},
       DOI = {10.1007/BF01019721},
       URL = {https://doi.org/10.1007/BF01019721},
}

@article {Kusuda,
    AUTHOR = {Kusuda, M.},
     TITLE = {A characterization of scattered {$C^*$}-algebras and its
              application to {$C^*$}-crossed products},
   JOURNAL = {J. Operator Theory},
  FJOURNAL = {Journal of Operator Theory},
    VOLUME = {63},
      YEAR = {2010},
    NUMBER = {2},
     PAGES = {417--424},
      ISSN = {0379-4024,1841-7744},
   MRCLASS = {46L55 (46L05)},
  MRNUMBER = {2651922},
MRREVIEWER = {Sriwulan\ Adji},
}

@article {LancienPerezGarcia,
    AUTHOR = {Lancien, C. and P\'erez-Garc\'ia, D.},
     TITLE = {Correlation length in random {MPS} and {PEPS}},
   JOURNAL = {Ann. Henri Poincar\'e},
  FJOURNAL = {Annales Henri Poincar\'e. A Journal of Theoretical and
              Mathematical Physics},
    VOLUME = {23},
      YEAR = {2022},
    NUMBER = {1},
     PAGES = {141--222},
      ISSN = {1424-0637,1424-0661},
   MRCLASS = {81P40 (81V70 82D03)},
  MRNUMBER = {4361873},
MRREVIEWER = {Mathieu\ Lewin},
       DOI = {10.1007/s00023-021-01087-4},
       URL = {https://doi.org/10.1007/s00023-021-01087-4},
}

@incollection {Lemm_ATMP,
    AUTHOR = {Lemm, M.},
     TITLE = {Finite-size criteria for spectral gaps in {$D$}-dimensional
              quantum spin systems},
 BOOKTITLE = {Analytic trends in mathematical physics},
    SERIES = {Contemp. Math.},
    VOLUME = {741},
     PAGES = {121--132},
 PUBLISHER = {Amer. Math. Soc., [Providence], RI},
      YEAR = {[2020] \copyright 2020},
      ISBN = {978-1-4704-4841-7},
   MRCLASS = {82B10},
  MRNUMBER = {4047784},
MRREVIEWER = {Angelo\ Lucia},
       DOI = {10.1090/conm/741/14923},
       URL = {https://doi.org/10.1090/conm/741/14923},
}

@article {Lentz_et_al,
    AUTHOR = {Lenz, D. and Peyerimhoff, N. and Veseli{\'c}, I.},
     TITLE = {Groupoids, von {N}eumann algebras and the integrated density
              of states},
   JOURNAL = {Math. Phys. Anal. Geom.},
  FJOURNAL = {Mathematical Physics, Analysis and Geometry. An International
              Journal Devoted to the Theory and Applications of Analysis and
              Geometry to Physics},
    VOLUME = {10},
      YEAR = {2007},
    NUMBER = {1},
     PAGES = {1--41},
      ISSN = {1385-0172,1572-9656},
   MRCLASS = {46L10 (35J10 46L51 46N55 47B80 81Q10 82B10 82B44)},
  MRNUMBER = {2340531},
MRREVIEWER = {Michael\ P.\ Lamoureux},
       DOI = {10.1007/s11040-007-9019-2},
       URL = {https://doi.org/10.1007/s11040-007-9019-2},
}

@article {MichalakisZwolak,
    AUTHOR = {Michalakis, S. and Zwolak, J. P.},
     TITLE = {Stability of frustration-free {H}amiltonians},
   JOURNAL = {Comm. Math. Phys.},
  FJOURNAL = {Communications in Mathematical Physics},
    VOLUME = {322},
      YEAR = {2013},
    NUMBER = {2},
     PAGES = {277--302},
      ISSN = {0010-3616,1432-0916},
   MRCLASS = {81Q15},
  MRNUMBER = {3077916},
MRREVIEWER = {Fumio\ Hiroshima},
       DOI = {10.1007/s00220-013-1762-6},
       URL = {https://doi.org/10.1007/s00220-013-1762-6},
}

@article {NabokoNicholsStolz,
    AUTHOR = {Naboko, S. and Nichols, R. and Stolz, G.},
     TITLE = {Simplicity of eigenvalues in {A}nderson-type models},
   JOURNAL = {Ark. Mat.},
  FJOURNAL = {Arkiv f\"or Matematik},
    VOLUME = {51},
      YEAR = {2013},
    NUMBER = {1},
     PAGES = {157--183},
      ISSN = {0004-2080,1871-2487},
   MRCLASS = {82B44 (35R60 60H25)},
  MRNUMBER = {3029341},
MRREVIEWER = {Constanza\ Rojas-Molina},
       DOI = {10.1007/s11512-011-0155-3},
       URL = {https://doi.org/10.1007/s11512-011-0155-3},
}

@article {Nachteragele,
    AUTHOR = {Nachtergaele, B.},
     TITLE = {The spectral gap for some spin chains with discrete symmetry
              breaking},
   JOURNAL = {Comm. Math. Phys.},
  FJOURNAL = {Communications in Mathematical Physics},
    VOLUME = {175},
      YEAR = {1996},
    NUMBER = {3},
     PAGES = {565--606},
      ISSN = {0010-3616,1432-0916},
   MRCLASS = {82B20 (47N55 82B10)},
  MRNUMBER = {1372810},
MRREVIEWER = {Christof\ K\"ulske},
       URL = {http://projecteuclid.org/euclid.cmp/1104276093},
}

@article {NachtergaeleOgataSims,
    AUTHOR = {Nachtergaele, B. and Ogata, Y. and Sims, R.},
     TITLE = {Propagation of correlations in quantum lattice systems},
   JOURNAL = {J. Stat. Phys.},
  FJOURNAL = {Journal of Statistical Physics},
    VOLUME = {124},
      YEAR = {2006},
    NUMBER = {1},
     PAGES = {1--13},
      ISSN = {0022-4715,1572-9613},
   MRCLASS = {82C10 (82B10 82C20)},
  MRNUMBER = {2256615},
MRREVIEWER = {H.\ Araki},
       DOI = {10.1007/s10955-006-9143-6},
       URL = {https://doi.org/10.1007/s10955-006-9143-6},
}

@article {NachtergaeleReschke,
    AUTHOR = {Nachtergaele, B. and Reschke, J.},
     TITLE = {Slow propagation in some disordered quantum spin chains},
   JOURNAL = {J. Stat. Phys.},
  FJOURNAL = {Journal of Statistical Physics},
    VOLUME = {182},
      YEAR = {2021},
    NUMBER = {1},
     PAGES = {Paper No. 12, 28},
      ISSN = {0022-4715,1572-9613},
   MRCLASS = {82C44 (82C10)},
  MRNUMBER = {4197413},
       DOI = {10.1007/s10955-020-02681-2},
       URL = {https://doi.org/10.1007/s10955-020-02681-2},
}

@incollection {NachtergaeleScholzWerner,
    AUTHOR = {Nachtergaele, B. and Scholz, V. B. and Werner,
              R. F.},
     TITLE = {Local approximation of observables and commutator bounds},
 BOOKTITLE = {Operator methods in mathematical physics},
    SERIES = {Oper. Theory Adv. Appl.},
    VOLUME = {227},
     PAGES = {143--149},
 PUBLISHER = {Birkh\"auser/Springer Basel AG, Basel},
      YEAR = {2013},
      ISBN = {978-3-0348-0530-8; 978-3-0348-0531-5},
   MRCLASS = {46L10 (46L53)},
  MRNUMBER = {3050163},
MRREVIEWER = {Chul\ Ki\ Ko},
       DOI = {10.1007/978-3-0348-0531-5\_8},
       URL = {https://doi.org/10.1007/978-3-0348-0531-5_8},
}

@article {NachtergaeleSims,
    AUTHOR = {Nachtergaele, B. and Sims, R.},
     TITLE = {Lieb-{R}obinson bounds and the exponential clustering theorem},
   JOURNAL = {Comm. Math. Phys.},
  FJOURNAL = {Communications in Mathematical Physics},
    VOLUME = {265},
      YEAR = {2006},
    NUMBER = {1},
     PAGES = {119--130},
      ISSN = {0010-3616,1432-0916},
   MRCLASS = {82B10 (82B20)},
  MRNUMBER = {2217299},
MRREVIEWER = {Nicolae\ Angelescu},
       DOI = {10.1007/s00220-006-1556-1},
       URL = {https://doi.org/10.1007/s00220-006-1556-1},
}

@article {NachtergaeleSimsYoung,
    AUTHOR = {Nachtergaele, B. and Sims, R. and Young, A.},
     TITLE = {Quasi-locality bounds for quantum lattice systems. {I}.
              {L}ieb-{R}obinson bounds, quasi-local maps, and spectral flow
              automorphisms},
   JOURNAL = {J. Math. Phys.},
  FJOURNAL = {Journal of Mathematical Physics},
    VOLUME = {60},
      YEAR = {2019},
    NUMBER = {6},
     PAGES = {061101, 84},
      ISSN = {0022-2488,1089-7658},
   MRCLASS = {81T25 (82B10 82B20)},
  MRNUMBER = {3964149},
MRREVIEWER = {Lech\ Jak\'obczyk},
       DOI = {10.1063/1.5095769},
       URL = {https://doi.org/10.1063/1.5095769},
}

@article {NachtergaeleSimsYoung_2,
    AUTHOR = {Nachtergaele, B. and Sims, R. and Young, A.},
     TITLE = {Quasi-locality bounds for quantum lattice systems. {P}art
              {II}. {P}erturbations of frustration-free spin models with
              gapped ground states},
   JOURNAL = {Ann. Henri Poincar\'e},
  FJOURNAL = {Annales Henri Poincar\'e. A Journal of Theoretical and
              Mathematical Physics},
    VOLUME = {23},
      YEAR = {2022},
    NUMBER = {2},
     PAGES = {393--511},
      ISSN = {1424-0637,1424-0661},
   MRCLASS = {81T25 (82B10 82B20)},
  MRNUMBER = {4386441},
       DOI = {10.1007/s00023-021-01086-5},
       URL = {https://doi.org/10.1007/s00023-021-01086-5},
}

@article {NachtergaeleSimsYoung_bulk,
    AUTHOR = {Nachtergaele, B. and Sims, R. and Young, A.},
     TITLE = {Stability of the bulk gap for frustration-free topologically
              ordered quantum lattice systems},
   JOURNAL = {Lett. Math. Phys.},
  FJOURNAL = {Letters in Mathematical Physics},
    VOLUME = {114},
      YEAR = {2024},
    NUMBER = {1},
     PAGES = {Paper No. 24, 55},
      ISSN = {0377-9017,1573-0530},
   MRCLASS = {82B20 (46L60 81Q10 81Q15 82B10)},
  MRNUMBER = {4700364},
       DOI = {10.1007/s11005-023-01767-8},
       URL = {https://doi.org/10.1007/s11005-023-01767-8},
}

@incollection {NachtergaeleVershyninaZagrebnov,
    AUTHOR = {Nachtergaele, B. and Vershynina, A. and Zagrebnov,
              V. A.},
     TITLE = {Lieb-{R}obinson bounds and existence of the thermodynamic
              limit for a class of irreversible quantum dynamics},
 BOOKTITLE = {Entropy and the quantum {II}},
    SERIES = {Contemp. Math.},
    VOLUME = {552},
     PAGES = {161--175},
 PUBLISHER = {Amer. Math. Soc., Providence, RI},
      YEAR = {2011},
      ISBN = {978-0-8218-6898-0},
   MRCLASS = {82C10 (46L57)},
  MRNUMBER = {2868047},
       DOI = {10.1090/conm/552/10916},
       URL = {https://doi.org/10.1090/conm/552/10916},
}

@article {Pastur75,
    AUTHOR = {Pastur, L. A.},
     TITLE = {Spectral properties of disordered systems in the one-body
              approximation},
   JOURNAL = {Comm. Math. Phys.},
  FJOURNAL = {Communications in Mathematical Physics},
    VOLUME = {75},
      YEAR = {1980},
    NUMBER = {2},
     PAGES = {179--196},
      ISSN = {0010-3616,1432-0916},
   MRCLASS = {81C10 (35P20 82A57)},
  MRNUMBER = {582507},
MRREVIEWER = {V.\ L.\ Girko},
       URL = {http://projecteuclid.org/euclid.cmp/1103908097},
}

@incollection {Pastur,
    AUTHOR = {Pastur, L.},
     TITLE = {On the thermodynamic limit for disordered spin systems},
 BOOKTITLE = {Geometric aspects of functional analysis},
    SERIES = {Lecture Notes in Math.},
    VOLUME = {1850},
     PAGES = {243--268},
 PUBLISHER = {Springer, Berlin},
      YEAR = {2004},
      ISBN = {3-540-22360-6},
   MRCLASS = {82B44 (82-02 82B20)},
  MRNUMBER = {2087162},
       DOI = {10.1007/978-3-540-44489-3\_19},
       URL = {https://doi.org/10.1007/978-3-540-44489-3_19},
}

@article {PasturFigotin,
    AUTHOR = {Pastur, L. A. and Figotin, A. L.},
     TITLE = {On the theory of disordered spin systems},
   JOURNAL = {Teoret. Mat. Fiz.},
  FJOURNAL = {Akademiya Nauk SSSR. Teoreticheskaya i Matematicheskaya
              Fizika},
    VOLUME = {35},
      YEAR = {1978},
    NUMBER = {2},
     PAGES = {193--210},
      ISSN = {0564-6162},
   MRCLASS = {82.60},
  MRNUMBER = {503349},
}

@article {Pettis,
    AUTHOR = {Pettis, B. J.},
     TITLE = {On integration in vector spaces},
   JOURNAL = {Trans. Amer. Math. Soc.},
  FJOURNAL = {Transactions of the American Mathematical Society},
    VOLUME = {44},
      YEAR = {1938},
    NUMBER = {2},
     PAGES = {277--304},
      ISSN = {0002-9947,1088-6850},
   MRCLASS = {28B05 (46G10)},
  MRNUMBER = {1501970},
       DOI = {10.2307/1989973},
       URL = {https://doi.org/10.2307/1989973},
}

@misc{RoonSchenker,
      title={Finitely Correlated States Driven by Topological Dynamics}, 
      author={Roon, E. B. and Schenker, J. H.},
      year={2025},
      eprint={2507.07287},
      archivePrefix={arXiv},
      primaryClass={math-ph},
      url={https://arxiv.org/abs/2507.07287}, 
    note = {arXiv:2507.07287},
}

@article {SimsStolz,
    AUTHOR = {Sims, R. and Stolz, G.},
     TITLE = {Many-body localization: concepts and simple models},
   JOURNAL = {Markov Process. Related Fields},
  FJOURNAL = {Markov Processes and Related Fields},
    VOLUME = {21},
      YEAR = {2015},
    NUMBER = {3},
     PAGES = {791--822},
      ISSN = {1024-2953},
   MRCLASS = {81Q80 (60H25 82C22)},
  MRNUMBER = {3494775},
MRREVIEWER = {Hatem\ Najar},
}

@incollection {Stolz_ALoc,
    AUTHOR = {Stolz, G.},
     TITLE = {An introduction to the mathematics of {A}nderson localization},
 BOOKTITLE = {Entropy and the quantum {II}},
    SERIES = {Contemp. Math.},
    VOLUME = {552},
     PAGES = {71--108},
 PUBLISHER = {Amer. Math. Soc., Providence, RI},
      YEAR = {2011},
      ISBN = {978-0-8218-6898-0},
   MRCLASS = {82B44 (81Q10)},
  MRNUMBER = {2868042},
MRREVIEWER = {Guido\ Gentile},
       DOI = {10.1090/conm/552/10911},
       URL = {https://doi.org/10.1090/conm/552/10911},
}

@incollection {Stolz_MBL,
    AUTHOR = {Stolz, G.},
     TITLE = {Aspects of the mathematical theory of disordered quantum spin
              chains},
 BOOKTITLE = {Analytic trends in mathematical physics},
    SERIES = {Contemp. Math.},
    VOLUME = {741},
     PAGES = {163--197},
 PUBLISHER = {Amer. Math. Soc., [Providence], RI},
      YEAR = {[2020] \copyright 2020},
      ISBN = {978-1-4704-4841-7},
   MRCLASS = {82B44},
  MRNUMBER = {4047786},
MRREVIEWER = {Chigak\ Itoi},
       DOI = {10.1090/conm/741/14925},
       URL = {https://doi.org/10.1090/conm/741/14925},
}

@article {Takeda,
    AUTHOR = {Takeda, Z.},
     TITLE = {Inductive limit and infinite direct product of operator
              algebras},
   JOURNAL = {Tohoku Math. J. (2)},
  FJOURNAL = {The Tohoku Mathematical Journal. Second Series},
    VOLUME = {7},
      YEAR = {1955},
     PAGES = {67--86},
      ISSN = {0040-8735,2186-585X},
   MRCLASS = {46.2X},
  MRNUMBER = {74800},
MRREVIEWER = {E.\ L.\ Griffin, Jr.},
       DOI = {10.2748/tmj/1178245105},
       URL = {https://doi.org/10.2748/tmj/1178245105},
}

\end{document}